\RequirePackage{fix-cm}
\documentclass[smallcondensed]{svjour3}  

\smartqed 
\usepackage{graphics}
\usepackage{graphicx}
\usepackage{subfigure}
\usepackage{verbatim}
\usepackage{epsfig}
\usepackage{xspace,amsfonts,amsmath}

\newcommand{\nodes}{\mathcal{N}}

\newcommand{\remove}[1]{}

\newcommand{\PR}[1]{\ensuremath{{\mathbf{Pr}\left[#1\right]}}}

\newcommand{\cF} {\ensuremath{\mathcal F}}
\newcommand{\cE} {\ensuremath{\mathcal E}}
\newcommand{\cI} {\ensuremath{\mathcal X_i}}
\newcommand{\cJ} {\ensuremath{\mathcal E_i}}
\newcommand{\cY}{\ensuremath{\mathcal Y_j}}

\newcommand{\cV} {\ensuremath{\mathcal E_j^v}}

\newcommand{\hcE} {\ensuremath{\bar{\mathcal E}}}
\newcommand{\hcF} {\ensuremath{\bar{\mathcal F}}}

\newtheorem{defn}{Definition}
\newtheorem{fact}{Fact}
\newtheorem{prop}{Proposition}
\newtheorem{observation}{Observation}

\newcommand{\bE}{{\mathbb E}}
\newcommand{\bP}{{\mathbb P}}

\newif\ifcomment\commentfalse

\def\commentOFF{\commentfalse}

\long\outer\def\bc#1\ec{{\ifcomment \sloppy  $[${\bf Alberto suggests}]
{{#1}} \textbf{[end]} \fi }}

\long\outer\def\br#1\er{{\ifcomment \sloppy  $[${\bf Alberto suggest remove}]
{{#1}} \textbf{[end]} \fi }}

\long\outer\def\bo#1\eo{{\ifcomment \sloppy  $[${\bf instead of}]
{\textit{#1}} \textbf{[end]}  \fi }}

\long\outer\def\BC#1\EC{{\ifcomment \sloppy \par \#  \dotfill
{\textsc{#1}} \dotfill \# \par \fi }}

\commentOFF

\begin{document} 

\title{Social-Aware Forwarding Improves Routing Performance in Pocket Switched Networks\thanks{A short version of this paper appeared in the Proceedings of the European Symposium on Algorithms (ESA) 2011.}}

\author{Josep D\'{i}az \and Alberto Marchetti-Spaccamela \and
 Dieter Mitsche \and Paolo Santi \and Julinda Stefa}
 
 \institute{Josep D\'{i}az\at
              Universitat Politecnica de Catalunya \\
              Jordi Girona Salgado 1-3, E-08034 Barcelona -- SPAIN\\
              \email{diaz@lsi.upc.edu} \and
              Alberto Marchetti-Spaccamela\at
              Sapienza Universit\'a di Roma \\
              Via Ariosto 25, 00185 Roma -- ITALY\\
              \email{alberto@dis.uniroma1.it} \and
              Dieter Mitsche\at
              Ryerson University \\
              350 Victoria Street, M5B2K3 Toronto -- Canada\\
              \email{dmitsche@ryerson.ca} \and     
              Paolo Santi\at
              IIT - CNR \\
              Via G. Moruzzi 1, 56124 Pisa -- ITALY\\
              Tel.: +39-050-3152411\\
              Fax: +39-050-3152333\\
              \email{paolo.santi@iit.cnr.it}  \and
              Julinda Stefa \at
              Sapienza Universit\'a di Roma\\
              Via Salaria 113, 00198 Roma -- ITALY\\
               \email{stefa@di.uniroma1.it}         
 }
\date{Received: date / Accepted: date}

%\numberofauthors{4}
%\author{}

\maketitle

\begin{abstract}
In this paper, we analyze the performance of two routing protocols for opportunistic networks, which are representative of social-oblivious and social-aware forwarding. In particular, we derive bounds on the expected message delivery time for a recently introduced stateless, social-aware forwarding protocol using interest similarity between individuals, and the well-known BinarySW protocol. 
We compare both from the theoretical and  experimental point of view the asymptotic performance of Interest-Based (IB) forwarding and BinarySW under two mobility scenarios, modeling situations in which pairwise meeting rates between nodes are either {\em independent of} or {\em correlated to} the similarity of their interests. 
We formally prove that, under the assumption that sender and destination of a message have orthogonal interests, IB forwarding provides asymptotically better performance than BinarySW with interest-correlated mobility. 
In fact, in this situation BinarySW yields {\em unbounded} expected delivery time, compared to {\em bounded} expected delivery time yielded by IB forwarding. On the other hand, when mobility of nodes is independent of their interests, the two forwarding approaches provide the same asymptotic performance. Our theoretical results are qualitatively confirmed by a simulation-based evaluation based on both a real-world trace and a synthetic (but realistic) mobility model. The analysis is then extended to consider less pessimistic hypothesis on similarity of sender and destination interests, to a model where the sender knows the ID of the destination but not its interests, and to forwarding approaches where multiple copies of the messages can travel along hop-bounded paths to destination. 

\end{abstract}

\section{Introduction}
Opportunistic networks, in which occasional communication opportunities between pairs or small groups of nodes are exploited to circulate messages, are expected to play a major role in next generation short range wireless networks  \cite{Resta,Spyro,Spyro2}. In particular, pocket-switched networks (PSNs) \cite{Hui3}, in which network nodes are individuals carrying around smart devices with direct wireless communication links, are expected to become widespread in a few years.
Message exchange in opportunistic networks is ruled by the {\em store-carry-}and{\em-forward} mechanism typical of delay-tolerant networks \cite{Fall}:  a node (either the sender, or a relay node) stores the message in its buffer and carries it around, until a communication opportunity with another node arises, upon which the message can be forwarded to another node (the destination, or another relay node). 

Given this basic forwarding mechanism, a great deal of attention has been devoted in past years to optimize the forwarding policy of routing protocols. \remove{More specifically, upon occasional encounter between two network nodes $A$ and $B$, the forwarding policy is in charge of deciding, for each message $M$ stored in $A$'s buffer, whether to copy $M$ into $B$'s buffer (and viceversa). }
Recently, several authors have proposed optimizing forwarding strategies for PSNs based on the observation that, being these networks composed of {\em individuals} characterized by a collection of social relationships, these social relationships can actually be reflected in the meeting patterns between network nodes. Thus, knowledge of the social structure underlying the collection of individuals forming a PSN can be exploited to optimize the routing  strategy, e.g., favoring message forwarding towards ``socially well connected" nodes. Significant performance improvement of social-aware  approaches over social-oblivious approaches has been experimentally demonstrated \cite{Daly,Hui,Li}. \remove{Routing approaches along this line are, e.g., \cite{Daly,Hui,Li}, where the authors use different ``social" metrics to improve end-to-end routing performance. For instance, in \cite{Daly} the authors use the notions of ``ego-centric betweenness" and ``social similarity"; in \cite{Hui}, the authors propose the use of a social ``centrality" metric; in \cite{Li}, the authors use a ``social similarity" metric locally computed from the history of past encounters. Recently, a social-based approach based on a notion of ``ego-centric betweenness" has been proposed also to optimize multicast performance \cite{Gao}.}

Most existing social-aware forwarding approaches hinge on the ability of {\em storing} information on the state of the network that can be used to attempt to predict future meeting opportunities \cite{Boldrini,Costa,Daly,Hui,Ioannidis,Li}. Examples of state information stored at the nodes are history of past encounters, portion of the social network graph, etc. On the other hand, socially-oblivious routing protocols  such as epidemic \cite{Vahdat}, two-hops \cite{Gross} and the class of Spray-and-Wait protocols \cite{Spyro}, do not require storing additional information in the node buffers, which are then exclusively used to store the messages circulating in the network. Thus, comparing performance of social-aware vs. social-oblivious forwarding approaches would require modeling node buffers, which renders the resulting network model very complex. If storage capacity on the nodes is not accounted for in the analysis, unfair advantage would be given to social-aware approaches, which extensively use state information. 
This explains why the fundamental question of whether social-aware forwarding is superior to social-oblivious forwarding {\em per se} (and not due to storage of extensive status information) has remained unaddressed so far.

In \cite{Mei2}, a {\em stateless}, social-aware forwarding approach has been presented; this approach is motivated by the observation  that individuals with similar interests  meet relatively more often than individuals with diverse interests \cite{McP}.  The definition of this Interest-Based forwarding approach (IB forwarding in the following) allows a {\em fair} comparison -- i.e., under the same conditions for what concerns usage of storage resources -- between social-aware and social-oblivious forwarding approaches in PSNs.

\medskip 
\noindent
{\bf Our contributions.}
The main goal of this paper is to present, for the first time to our best knowledge, a comparison of asymptotic performance provided by social-aware and social-oblivious forwarding protocols for PSNs. For the reasons described above, we choose IB forwarding as a representative example of social-aware protocols, and BinarySW as a representative example of social-oblivious protocols. BinarySW \cite{Spyro} is chosen since in the mentioned work it is shown to be optimal within the class of Spray-and-Wait forwarding protocols, and given the extensive simulation-based evidences of its superiority within the class of stateless, social-oblivious approaches. 
Our interest in asymptotic investigation is motivated by the fact that PSN size can easily grow up to several thousands of nodes. 

The two protocols are compared under two different scenarios for what concerns node mobility: one, called {\em interest-based mobility}, in which mobility of individuals is influenced by similarity of their interests; and the second, called {\em social-oblivious mobility} in which mobility of individuals is oblivious to similarity of their interests. 
%To the best of our knowledge, the notion of interest-based mobility -- although empirically verified in \cite{Mei2,Noulas} -- has never been formalized in the literature. 
%Ours is then the first analytical investigation of opportunistic network performance in presence of a form of social-aware mobility.
%
%In other words, the approaches introduced in \cite{Daly,Hui,Li} are {\em stateful approaches.} When comparing social-aware with social oblivious routing strategies,  the effect of limited storage capacity at the nodes should be carefully taken into account. \remove{In fact, in presence of bounded memory size, using storage for keeping a significant amount of state information significantly reduces the amount of memory that can be used to buffer messages circulating in the network, with a negative effect on routing performance.} Unfortunately, experimental performance investigations of existing social-aware routing approaches \cite{Daly,Hui,Li} do not account for limited node memory size, thus biasing the comparison. 
%
%
%The  overall goal of this paper is to provide a {\em fair} comparison between the routing performance achieved by simple  social-aware and social oblivious forwarding strategies, differently from all social-aware routing approaches introduced so far.
%

The specific technical contributions of this paper are:

\begin{enumerate}
\item An asymptotic analysis of IB and BinarySW forwarding performance -- expressed in terms of expected message delivery time -- in case of both  interest-based and social-oblivious mobility. We consider the case when only one relay node can be used to speed up message delivery and we  prove, under reasonable probabilistic assumptions, that  {\em IB forwarding provides asymptotic performance benefits compared to BinarySW}:  IB forwarding yields {\em bounded} expected message delivery time under both mobility models,  while BinarySW yields {\em bounded} expected delivery time with social-oblivious node mobility, but {\em unbounded} delivery time with interest-based mobility. The result  that IB forwarding provides an asymptotic performance gap with respect to BinarySW forwarding with interest-based mobility might not be surprising. However, ours is the first formal proof of this asymptotic performance benefit.
\item We quantitatively confirm the analysis of 1) through simulations  based both on a real-world data trace and a synthetic   human mobility model recently introduced in \cite{Mei}.
\item We extend the analysis of 1) in several ways. First, we consider the case when many relay nodes, more copies of the message, and more hops can be used to speed up message delivery. We show that the expected delivery time of BinarySW with interest-based mobility is asymptotically the same, thus proving that the asymptotic performance benefit of IB vs. BinarySW forwarding is retained also under these more general conditions.
We also consider a version of the forwarding algorithm in which the sender  knows the ID of the destination, but it does not know its {\em interest profile} (see next section for a formal definition of interest profile). We show that the expected message delivery time with IB forwarding and interest-based mobility remains bounded even in this more challenging networking scenario if we allow a limited number of relay nodes. 
\item The analysis of 1) and 3) is done under the scenario in which source and destination of a message have orthogonal interests. We also consider an average-case scenario in which the angle between the vectors representing source and destination interests is uniformly distributed in $[0,\pi/2]$, and show that under these less pessimistic conditions BinarySW yields {\em bounded} expected delivery time -- i.e., the same asymptotic performance as IB forwarding -- also with interest-based mobility.
\item A byproduct of the above analysis is the definition of a simple model of pair-wise contact frequency correlating similarity of individual interests with their meeting rate. We believe this model might be useful in studying other social-related properties of PSNs, and we deem such model a contribution in itself.
\end{enumerate}

The rest of this paper is organized as follows. In the next section, we shortly survey related work. In Section \ref{model}, we present the network and mobility models, and the forwarding approaches considered in this paper. In Section \ref{WorstCase}, we present the analysis of forwarding performance with social-oblivious mobility, while Section \ref{IBanalysis} considers the case of interest-based mobility. We will then present simulation results supporting the main theoretical findings of sections \ref{WorstCase} and \ref{IBanalysis} in Section \ref{sims}. In Section \ref{MoreHops} we extend the analysis to the case of multiple copies of the message circulating in the network, and arbitrary length of the message delivery path. In Section \ref{Unif}, we consider the case in which source and destination of a message do not have orthogonal interests. Finally, Section \ref{conclusions} concludes the paper.

\section{Related work}

Performance analysis of opportunistic networks has been subject of intensive research in recent years. In particular, the analysis of routing performance -- expressed in terms of the expected message delivery time, as done in this paper -- has been considered in \cite{AlH1,Chain,Groene,Spyro,Spyro2,Zhang}. More recently, also the distribution of the message delivery time has been studied \cite{Resta}. These studies assume a mobility model equivalent to one of the two-mobility models considered in this paper, namely the social-oblivious mobility model. Furthermore, they all consider social-oblivious routing protocols such as epidemic \cite{Vahdat}, two-hops \cite{Gross}, and BinarySW routing \cite{Spyro2}.

Recently, several opportunistic networking protocols accounting for social relationships between network members have been proposed. These protocols encompass different networking primitives such as unicast \cite{Daly,Hui,Li}, multicast \cite{Gao}, and publish-subscribe services \cite{Boldrini,Costa,Ioannidis}. While superiority of social-aware approaches over social-oblivious ones has been established in the literature based on several simulation-based evaluations, to our best knowledge theoretical analysis of social-aware networking protocols for opportunistic networks has remained unaddressed so far. As commented in the Introduction, this is likely due to the fact that existing social-aware forwarding protocols heavily build upon a notion of network state locally stored at the nodes to improve performance, hence a theoretical evaluation of their performance would require including in the model the evolution of the network state and/or buffer occupancy at the nodes, which appears to be a very difficult task.

In a recent paper \cite{Mei2}, some of the authors of this paper proposed a social-aware forwarding approach for opportunistic networks which, for the first time, does not exploit local storage of network state to speed up the forwarding process. Instead, the approach is based on a notion of similarity of interests between individuals, and on the empirical observation that individuals with relatively similar interests tend to meet more often than individuals with relatively diverse interests. 
In this paper, we take advantage of the stateless feature of the recently proposed social-aware forwarding approach of \cite{Mei2}, and present for the first time a theoretical investigation of social-aware forwarding protocols in opportunistic networks.

A challenging issue when investigating performance of social-aware forwarding protocols is taking into account the social dimension in the mobility model used to analyze routing performance. To the best of our knowledge, while different assumptions about pair-wise inter-meeting rates have been made in the literature (such as exponential \cite{AlH1,Groene,Spyro,Spyro2,Zhang}, power law \cite{Chain}, and power law with exponential tail \cite{Kara}), all existing analyses share the common feature that the pair-wise meeting rates between {\em any} pair of nodes in the network have the same stochastic property (e.g., they are all exponential random variables with a fixed rate $\lambda$ \cite{AlH1,Groene,Spyro,Spyro2,Zhang}). Clearly, these models cannot be used to express the influence of social relationships on pairwise meeting rates since, independently of the specific stochastic assumptions, the stochastic process modeling pair-wise meeting events between nodes is oblivious to node identities. A major contribution of this paper is introducing, for the first time in the literature, a simple model of pair-wise meeting rates which accounts for social relationships between each specific pair of nodes in the network. In particular, inspired by the notion of interest space introduced in \cite{Mei2}, we use similarity between user interest profiles as a proxy of the intensity of their social relationships, and define the intensity of the meeting process between any two specific nodes $A$ and $B$ in the network to be proportional to the similarity of their interest profiles (see the following for details). The pair-wise meeting process between any two nodes $A$ and $B$ is assumed to have exponential distribution, which is representative at least of the tail of inter-meeting time distributions extracted from real-world traces \cite{Kara}. We stress that ours is the first model of pair-wise meeting rates explicitly accounting for a form of social relationships between individuals; in particular, the rate of the exponential random variable modeling meeting rate between any two nodes $A$ and $B$ is a function of $A$'s and $B$'s interest profiles.

\section{The Network and Mobility Models}\label{model}
We consider a network of $n+2$ nodes, which we denote 
$\nodes=\{S,D, R_1,\ldots ,$ $R_n\}$: a {\em source node} $S$, a {\em destination node} $D$, and $n$ potential {\em relay nodes} $R_1,\dots,R_n$. 
Following the  model presented in  \cite{Mei2}, we model each of the $n+2$ nodes as 
a point in an $m$-dimensional {\em interest space} $[0,1]^m$, where  $m$ is the total number of interests  and $m\ll n$. We assume $m=\Theta(1)$. The $m$-dimensional vector associated with a node defines its {\em interest profile}, i.e., its degrees of interest in the various  dimensions of the interest space. Each node $A\in \nodes$ is thus assigned an $m$-dimensional vector $A[a_1,\ldots ,a_m]$ in the interest space. As in  \cite{Mei2}, we use the well-known {\em cosine similarity} metric \cite{Deza}, which measures similarity between two nodes
$A$ and $B$  as $\cos(\angle (AB))$, the cosine of the angle formed by $A$ and $B$. Since the cosine similarity metric implies that the norm of the vectors is not relevant, we can consider all vectors to have unit norm. 

The previous model is equivalent to assume nodes are represented as points in the  positive orthant of the $m$-dimensional unit sphere $\mathcal{S}$. Moreover, we assume all interests to be non-negative. Therefore,   $0\le \cos(\alpha_{AB})\le 1$, with higher values of $\cos(\angle(AB))$ corresponding to a higher similarity in interests between $A$ and $B$. Thus, the cosine similarity metric can be used as a measure of the degree of ``homophily" -- similarity in interests and habits~\cite{McP} -- between individuals.
We assume $S$ and $D$ to have orthogonal interests, namely
$S[1,0,\ldots, 0]$, and $D[0,1,\ldots ,0]$. We call this scenario the {\em worst-case delivery scenario} since it corresponds to the worst case situation (i.e., a situation resulting in the largest expected delivery delay) under the interest-based mobility model -- see below for a formal definition of interest-based mobility. Furthermore, in the analysis below, we assume the following concerning the distribution of interest profiles in the interest space: first, the angle $\alpha_i$ between the $i$-th interest profile and $S$'s interest profile is chosen uniformly at random in $[0,\pi/2]$; then, from all unit vectors in the intersection of the positive orthant of the $m$-dimensional sphere with that $(m-1)$-dimensional subspace, one vector is chosen uniformly at random -- see Figure \ref{sphere}.

It is important to observe that, while nodes are assumed to move around according to some mobility model $\mathcal{M}$, node coordinates in the interest space {\em do not change over time}. This is coherent with what happens in real world, where individual interests change at a much larger time scale (months/years) than that needed to exchange messages within the network. Thus, when focusing on a single message delivery session, it is reasonable to assume that node interest profiles correspond to fixed points in the interest space.

Similar to most analytical works on opportunistic networks \cite{Resta,Spyro,Spyro2}, we do not make any assumption about nodes following a specific mobility model. Rather, we make assumptions about the meeting rates between individuals in the network. 
In particular, we assume that the mobility metric relevant to our purposes is the {\em expected meeting time}, which is formally defined as follows:
\begin{defn}
Let $A$ and $B$ be nodes  in the network, moving in a bounded region $R$ 
according to a mobility model $\mathcal{M}$. Assume that at  time $t=0$ both $A$ and $B$ are independently distributed in $R$ 
according to the stationary node spatial distribution of  $\mathcal{M}$,\footnote{It is well-known that some mobility models, such as RWP, give rise to a non-uniform node spatial distribution in stationary conditions.} and that $A$ and $B$ have a fixed transmission range. The {\em first meeting time} $T$ between  $A$ and $B$ is the random variable (r.v.) corresponding to the time interval elapsing between $t=0$ and the instant of time where $A$ and $B$ first come into each other's transmission range. The {\em expected meeting time} is the expected value of the r.v. $T$. 
\end{defn}

Following the literature \cite{Resta,Spyro,Spyro2}, we assume  the meeting time between any pair of nodes $A$ and $B$ is described by a Poisson point process  of intensity $\lambda_{AB}$, i.e., 
$T_{AB}$ follows an exponential distribution with parameter $(\lambda_{AB})$ and thus $\bE[T_{AB}]=\frac{1}{\lambda_{AB}}$. As mentioned in the previous section, we are aware that recent findings indicate that pair-wise meeting patterns obey a power law+exponential tail dichotomy \cite{Kara}. However, the simplifying assumption of exponentially distributed pair-wise meeting process is made in the following to reduce the complexities brought in the analysis by the assumption of social-aware meeting rates.

In the analysis, we use the following well-known properties of exponentially distributed random variables:
\begin{fact}\label{fact:sum_poisson}
Given a set of $n$ independent exponentially distributed random variables $X_1,\ldots,X_n$ with parameters $\lambda_1,\ldots,\lambda_n$, let $X_m=\min\{X_1,\dots,X_n\}$ denote the first order statistic of the $n$ variables. Then, $X_m$ is an exponentially distributed random variable with rate parameter $\lambda_m=\sum_{i=1}^n \lambda_i$.
\end{fact}

\begin{fact}\label{fact:index}
Given $X_1,\ldots,X_n$ and $X_m$ as above, for each $j=1,\dots,n$,
\[
Prob(X_m=X_j)=\frac{\lambda_j}{\sum_{i=1}^n\lambda_i}~.
\]
\end{fact}

In the sequel we consider  two mobility models and two forwarding algorithms.  The 
{\em social-oblivious} and {\em interest-based} mobility models are defined as follows:
\begin{itemize}
\item[--] {\em social oblivious mobility}: for any $A,B\in\nodes$, the meeting rate is $\lambda_{AB}=\lambda$ for some $\lambda>0$, independent of $A$ and $B$. This corresponds to the situation in which node mobility is not influenced by the social relationships between $A$ and $B$, and it is the standard model used in opportunistic network analysis \cite{Resta,Spyro,Spyro2}. 
\item[--]  {\em  interest-based mobility}: the meeting rate $\lambda_{AB}$ between $A$ and $B$ is defined as $\lambda_{AB}=k\cdot \cos(\alpha_{AB})+\delta(n)$. The first term in the definition of $\lambda_{AB}$ accounts for the ``homophily degree" between individuals $A$ and $B$, introducing a positive correlation between ``homophily degree" and frequency of meetings. The second term instead accounts for the fact that occasional meetings can occur also between perfect strangers;  we are  interested in the case $\delta(n)\to 0$  as $n \rightarrow \infty$, which corresponds to the fact that as $n$ grows, the probability of meeting by chance a specific individual decreases. Finally,  $k>0$ is a parameter modeling the intensity of the interest-based mobility component. 
\end{itemize}
\medskip
We are interested in characterizing the performance of {\em routing algorithms}, i.e., the dynamics related to delivery of a message $M$ from $S$ to $D$. With a slight abuse of notation, we use $S$, $D$, or $R_i$ to denote both a node, and its coordinates in the interest space. The dynamics of message delivery is governed by a routing protocol, which determines how many copies of $M$ shall circulate in the network, and the forwarding rules. 
In our analysis, we consider instances of both, social-aware and social-oblivious forwarding rules. More specifically, we consider the following two routing strategies when sending the message $M$ from $S$ to $D$:
%As commented in the Introduction, most existing social-aware routing protocols  are {\em stateful} approaches.  
%In order to isolate the effects of different forwarding decisions (social-aware or social-oblivious) on routing performance from the contribution of maintaining a state, in our analysis we consider a simple {\em stateless}, {\em social-aware} routing protocol recently introduced in \cite{Mei2}, which we rename here {\em Interest-Based}(IB) routing for convenience.

% IB routing builds upon the assumption, qualitatively observed in sociological studies \cite{McP}, that individuals with common interests tend to meet with each other more often than with individuals with diverse interests. This observation has been recently validated on real-world traces in \cite{Mei2,Noulas}. Interest-based information dissemination has been used in \cite{Costa} to improve performance of a publish subscribe system. However, the approach of \cite{Costa} is stateful, hence cannot be directly used to our purposes. 

\begin{itemize}
\item {\em  FirstMeeting} (FM): $S$ is allowed to generate two copies of $M$; $S$ always keeps a copy of $M$ for itself. Let $R_j$ be the first node met by $S$ amongst nodes $\{R_1,\dots,R_n\}$. If $R_j$ is met before node $D$, the second copy of $M$ is delivered to node $R_j$. From this point on, no new copy of the message can be created nor transferred to other nodes, and $M$ is delivered to $D$ when the first node among $S$ and $R_j$ gets in touch with $D$. If $D$ is met by $S$ before any of the $R_i$'s, $M$ is delivered directly. This  protocol is  equivalent to BinarySW as defined in \cite{Spyro} when the number of message copies in BinarySW is fixed to 2. However, for convenience in the following we retain the name FM to describe this routing approach.
\item {\em  InterestBased} \cite{Mei2}: IB($\gamma$) routing is similar to FM, the only difference being that the second copy of $M$ is delivered by $S$ to the first node $R_k\in\{R_1,\dots,R_n\}$ met by $S$ such that $cos(R_k,D)\ge\gamma$, where $\gamma\in[0,1]$ is a tunable parameter. Note that IB(0) is equivalent to FM routing. If it happens that after time $n$ still no node in $\{R_1,\dots,R_n\}$ satisfying the forwarding condition is encountered, then the first relay node meeting $S$ after time $n$ is given the copy of $M$ independently of similarity between interest profiles. 
\end{itemize}

Note that implicit in the IB routing approach is the fact that a node $S$ generating a message $M$ for a certain destination node $D$ knows $D$'s interest profile. Conceptually, this is equivalent to the standard assumption that $S$ knows $D$'s address when sending message  $M$. Thus, we can think of $D$'s interest profile as her/his address in the network, although technically speaking, a node's interest profile cannot be directly used as address since uniqueness of node IDs in principle cannot be guaranteed. In Section \ref{MoreHops}, we will extend the analysis to cover the case in which $D$'s interest profile is not known to node $S$.

We remark that IB routing is a stateless approach:  interest profiles of encountered nodes are stored only for the time needed to locally compute the similarity metrics, and discarded afterwards. 
%Although stateless, IB routing requires storing a limited amount of extra information in the node's memory besides messages: the node's interest profile, and the interest profile of the destination for each stored message. However, note that interest profiles can be compactly represented using a number of bits which is independent of the number $n$ of network nodes, i.e., the additional amount of storage requested on the nodes is $O(1)$.
%On the other hand, stateful approaches such as \cite{Boldrini,Costa,Daly,Hui,Ioannidis,Li} require storing an amount of information which is at least proportional to the number of nodes in the network, i.e., it is $O(n)$ (in some cases it is even $O(n^2)$). Thus, comparing IB routing with a socially-oblivious routing protocol can be considered fair (in an asymptotic sense) from the viewpoint of storage capacity. 
Based on this observation, in the following we will make the standard assumption that node buffers have unlimited capacity \cite{Resta,Spyro,Spyro2}, which contributes to simplifying the analysis.

In the following, we denote by $T^{\mu}_{X}$  the random variable corresponding to the time at which  $M$ is first delivered to $D$,  assuming a routing protocol $X\in\{$FM,IB($\gamma$)$\}$ and a mobility model $\mu\in\{so,ib\}$, where $so$ and $ib$ represent social-oblivious and interest-based mobility, respectively. Our interest in characterizing delivery delay is due to the fact that, once a notion of TimeToLive is associated with a message, delivery delay can be used also to estimate the percentage of messages successfully delivered to destination.

For both algorithms and both mobility models, we  consider the following random variables: $T_1$ is the r.v. counting the time it takes for $S$ to meet the first node in the set $\mathcal{R}=\nodes \setminus \{S\}$;  $T_2$ is $0$ if $D$ is the first node in  $\mathcal{R}$ met by $S$; otherwise, if $R_j$ is the first relay node met by $S$, $T_2$ is the r.v. counting the time, starting at $T_1$, until the first node amongst $S$ and $R_j$ meets $D$.
\section{Bounds on the expected delivery time: Social Oblivious mobility}\label{WorstCase}
In this section, we evaluate FM and IB routing in the social-oblivious mobility scenario, giving asymptotic expressions for  $\bE[T^{so}_{FM}]$
and $\bE[T^{so}_{IB(\gamma)}]$. Our results prove that in the social oblivious mobility scenario  $\bE[T^{so}_{FM}]$
and $\bE[T^{so}_{IB(\gamma)}]$ are asymptotically equal to a constant.
\subsection{First-meeting routing}\label{so-mobility} 
For the sake of completeness, we include the derivation of $\bE[T^{so}_{FM}]$ under social-oblivious mobility, which can be easily done along the lines in \cite{Spyro}.
Since the meeting processes between $S$ and any other node are independent, and they can be modeled as exponentially distributed random variables with the same rate parameter $\lambda$, by Fact~\ref{fact:sum_poisson}, the time for node $S$ to meet the first node in set $\mathcal{R}=\{D,R_1,\dots,R_n\}$ is itself an exponentially distributed r.v. with rate parameter $(n+1)\lambda$. Thus, $\bE[T_1]=\frac{1}{(n+1)\lambda}$.
With probability $1/(n+1)$, the first node met by $S$ is $D$, and $M$ is delivered to destination. On the other hand, with probability $n/(n+1)$, starting from time $T_1$ we have  nodes $S$ and $R_j$ carrying a copy of message $M$. Identical   argument yields,
$\bE[T_2]=\frac{1}{2\lambda}.$ 
Putting everything together, we can conclude that: 
\[
\bE[T^{so}_{FM}]=\frac{1}{\lambda(n+1)}\frac{1}{n+1}+\left(\frac{1}{2\lambda}+\frac{1}{\lambda(n+1)}\right)\frac{n}{n+1},
\]
which converges to $\frac{1}{2\lambda}$, as $n \rightarrow \infty$. That is, the expected message delivery time with FM routing and social-oblivious mobility in very large networks converges to a {\em positive constant}. We summarize this result in the following proposition:
\begin{prop}\label{FMSO}
$\bE[T^{so}_{FM}] =\frac{1}{2\lambda}(1+o(1))$.
\end{prop}

\subsection{Interest-based routing}\label{IB}
We now consider the case of IB routing. For clarity of the exposition, we  set $\gamma=\frac{0.29}{m-1}$ and prove all results for this value of $\gamma$. The extension to other values of 
$\gamma \in (0,1)$ is straightforward.

We start with a technical lemma that will be also used in other parts of the paper.
Denote by $N_f$ the random variable counting the number of nodes $R_i$ satisfying the forwarding condition, i.e., nodes whose interest profile makes an angle of at most $\arccos \gamma$ with $D$'s interest profile. Furthermore, we add the condition that $R_i$ angle with $S$'s interest profile is at most $\arccos \frac{3\pi}{8}$. 
\begin{lemma}\label{lem:N_f}
 Denote by $N_f$ the random variable counting the number of nodes $R_i$ satisfying the following conditions: 1) $R_i$'s interest profile makes an angle at most $\arccos \gamma$ with $D$'s interest profile; and 2) $R_i$'s interest profile makes an angle of at most $\arccos \frac{3\pi}{8}$ with $S$'s interest profile. Then, 
with probability at least $1-e^{-\Theta(n)}$, for any $\nu > 0$, $N_f \geq (1-\nu)n\frac{1}{4(m-1)}$.
\end{lemma}
\begin{proof}
Since with probability $\frac14$, the angle between $S$ and an arbitrary intermediate node $R_i$ is between $\frac{\pi}{4}$ and $\frac{3\pi}{8}$, with that probability $\cos \angle(S,R_i) \leq 1/\sqrt{2}$. Since we assumed $S$ to have coordinates $(1,0,\ldots,0)$, this means that in that case $ \sum_{j \neq 1} (R_i)_j \geq \sum_{j \neq 1}((R_i)_j)^2 \geq 1-\frac{1}{\sqrt{2}} \geq 0.29$. Since among all positions making the same angle with $S$ all have the same probability to occur, with probability at least $\frac{1}{m-1}$, the value at the second coordinate is at least a $\frac{0.29}{m-1}=\gamma$. Thus, with probability at least $\frac{1}{4(m-1)}$, $\cos \angle(R_i,D) \geq \gamma$, or equivalently, $\angle(R_i,D) \leq \arccos \gamma$. Hence, $\bE[N_f] \geq n\frac{1}{4(m-1)}$. Since the positions of all nodes $R_i$ are chosen independently, by Chernoff bounds, for any $\nu > 0$, $\bP[N_f \leq (1-\nu)\bE[N_f]] \leq e^{-\Theta(n)}$, and the statement follows.
\end{proof}

\begin{prop}\label{IBSO}
$\bE[T^{so}_{IB(\gamma)}] =\frac{1}{2\lambda}(1+o(1))$, where $\gamma=\frac{0.29}{m-1}$.
\end{prop}
\begin{proof}
%\BC I suggest to put Sketch of proof  \EC
Define $N_f$ as above. By Lemma~\ref{lem:N_f}, %below (which requires an additional condition on the angles not needed here, but since this affects only $T_1$ which is asymptotically not important, we can apply it), 
 with probability at least $1-e^{-\Theta(n)}$, for any $\nu > 0$, $N_f \geq (1-\nu)n\frac{1}{4(m-1)}$. Denote by $\cE$ the event that $N_f$ has at least this size. Rewrite $\bE[T_1]$ as follows: $\bE[T_1]=\bE[T_1 |\cE]\bP[\cE]+\bE[T_1|\hcE]\bP[\hcE]$, where $\hcE$ is the complementary of event $\cE$. Since $\bE[T_1]\leq n(1+\frac{1}{\delta n})$, we have that
  $\bE[T_1 |\cE]$ is the dominating contribution. The fact that  $\bE[T_1 |\cE]= O(1/n)$, combined with the fact that $\bP[\cE]\ge 1-e^{-\Theta(n)}$ give $\bE[T_1]= O(1/n)$. As in the analysis of the FM routing algorithm, $\bE[T_2]=\frac{1}{2\lambda}$, and thus  $\bE[T^{so}_{IB}]=\frac{1}{2\lambda}(1+o(1))$.
 \end{proof}

Propositions~\ref{FMSO} and \ref{IBSO} imply the announced result $\bE[T^{so}_{FM}] =\bE[T^{so}_{IB(\gamma)}](1+o(1)) ={2\lambda}(1+o(1))$.

%%%%%%%%%%%%%%%%%%%%%%%%%%%%%
\section{Bounds on the expected delivery time: Interest-based mobility}\label{IBanalysis}
%Let us now consider IB routing. For clarity in the exposition, we fix the value $\gamma:=\frac{0.29}{m-1}$ and prove all results for this value of $\gamma$ (the extension to other values of $\gamma \in (0,1)$ is conceptually easy).

The analysis of FM and IB routing under interest-based mobility is more challenging than the one for the social-oblivious model, and the results clearly differentiate the asymptotic behavior of the two routing protocols. 
\subsection{First-meeting routing}
Consider now the case that FM routing is used in presence of interest-based mobility. The difficulty in performing the analysis stems from the fact that, under interest-based mobility, the rate parameters of the exponential random variables representing the first meeting time between $S$ and the nodes in the set $\mathcal{R}$ are themselves random variables. 

Denote by $\alpha_i$ the random variable representing the angle between node $S$ and $R_i$ in $\mathcal{S}$, and by $\lambda_i=k \cos \alpha_i+\delta$ the random variable corresponding to the meeting rate between $S$ and $R_i$. Recall that we assume that $S$ and $D$ are orthogonal, and that the $\alpha_i$s are distributed uniformly at random. Hence, the probability density for any $\alpha_i$ to attain any value $x \in [0,\pi/2]$ is $2/\pi$.

In order to make results in case of social oblivious and interest-based mobility comparable, we first derive the expected value of $\lambda_i$, and set the normalization constant $k$ in such a way that $\bE[\lambda_i]=\lambda$.
 We have
\begin{equation}\label{expectation}
\bE[\lambda_i]= \int_0^{\pi/2} \frac{2}{\pi} (k\cos(\alpha)+\delta) d\alpha = \frac{2k}{\pi}+\delta,
\end{equation}
and thus $k=\frac{\pi}{2}(\lambda-\delta)$.

To compute $\bE[T_1]$ exactly, we have to consider an $n$-fold integral taking into account all possible positions of the nodes $R_1,\ldots,R_n$ in the interest space\footnote{Recall that we are considering  the fixed, but randomly chosen, position of a node's interest profile in the interest space, not its physical position, which depends on the mobility model $\mathcal{M}$.}. As we will see shortly, $T_1$ is asymptotically negligible compared with $T_2$, therefore we can use the trivial lower bound of $\delta$ on the rate of the random variables corresponding to the first meeting time between $S$ and any other node, and thus we get $\bE[T_1] \leq \frac{1}{n\delta}$. For the same above described reason for the computation of $T_1$, computing $T_2$ exactly also seems difficult. In the following lemma, we give a lower bound on $\bE[T_2]$.
\begin{lemma}\label{lem:T2}
Under the above assumptions, for some constants $c, c' >0$, we have
\[
\bE[T_2] \geq \min\{c \log(1/\delta),\frac{c'n}{\log n}\}.
\]
\end{lemma}
\begin{proof}
Assume that $\delta=\omega(\log n/n)$. If $\delta$ is smaller, we can couple the model with some other $\hat{\delta}=\omega(\log n/n))$ and  obtain a new stochastic process whose meeting time $T_2$ is stochastically bounded from above by the meeting time $T_2$ of the original model. Let $R_f$ be the first node met by $S$, and assume $R_f\neq D$. We analyze the intensity  of the first meeting process between $R_f$ and $D$. Since this intensity is always greater than or equal to the intensity of the corresponding process between $S$ and $D$, we will have a bound for $T_2$. Recall that $S$ and $D$ are orthogonal in the interest space, hence they have minimal pairwise meeting rate.

 Partition the interval $[0,\pi/2]$ into subintervals $I_1,\ldots,I_{\pi/(2\delta)}$ of length $\delta$. Denote by $X_i$ the random variable corresponding to the number of points in the $i$-th subinterval. For any fixed $i$, $\bE[X_i]=\frac{n 2\delta}{\pi} =\omega(\log n)$. Using Chernoff's bounds, for any $\epsilon > 0$,  $\bP[X_i \leq (1-\epsilon)\bE[X_i]] \leq n^{-100}$, and $\bP[X_i \geq (1+\epsilon)\bE[X_i]] \leq n^{-100}$. Taking a union bound over all $\Theta(1/\delta) \leq n$ intervals, we see that with probability at least $1-n^{-98}$, the above property holds in {\em all} subintervals.
 
\noindent Consider now the random variable $\lambda_m=\sum_{i=1}^n\lambda_i$, corresponding to the rate parameter of the r.v. representing the first meeting time between node $S$ and $R_f$. From equation (\ref{expectation}), together with linearity of expectation we have $\bE[\lambda_m]=n(\frac{2k}{\pi}+\delta)$, and by Theorem A.1.15 of~\cite{AlonSpencer}, 
$\bP[\lambda_m \geq (1+\nu)n(\frac{2k}{\pi}+\delta)]=\bP[\lambda_m \geq (1+\nu)\bE[\lambda_m]] \leq n^{-100}$. Thus, with probability at least $1-n^{-100}$, the rate parameter $\lambda_m$ is at most $(1+\nu)n(\frac{2k}{\pi}+\delta)$. Hence, with probability at least $1-n^{-97}$, in {\em all} subintervals of length $\delta$ the number of nodes $X_i$ is within $(1 \pm \epsilon)\bE[X_i]$, and the rate parameter $\lambda_m$ is at most $(1+\nu)n(\frac{2k}{\pi}+\delta)$. Since in this lemma we are only interested in a lower bound on $\bE[T_2]$, we condition now under this event, call it $\cF_1$.  Observe that the rate parameter of the first meeting r.v. between a node in the $i$-th subinterval and $S$ is at least $k \cos(i \delta)+\delta$. Let $I_i$ denote the set of rate parameters belonging to the $i$-th sub-interval of $[0,\pi/2]$, and let $\lambda_{mi}=\sum_{j\in I_i} \lambda_j$. Applying again Theorem A.1.15 of~\cite{AlonSpencer}, for each subinterval $i$ with probability at least $1-n^{-100}$, we have:
\[
\lambda_{mi}\ge(1-\eta)(1-\epsilon)\frac{n \delta 2}{\pi}(k \cos(i \delta)+\delta)~,
\] 
where $\eta$ and $\epsilon$ are arbitrarily small positive numbers, and we also condition on this event, call it $\cF_2$. 
Hence, by Fact \ref{fact:index}, conditioned on $\cF_1 \wedge \cF_2$, the probability that node $R_f$ belongs to the $i$-th subinterval is at least
\[
\frac{2\delta n (1-\eta)(1-\epsilon)(k\cos (i \delta)+\delta)/\pi}{(1+\nu)n(2k/\pi+\delta)}~.
\]
Observe also that  if a node belongs to the $i$th subinterval, then the rate parameter of the r.v. corresponding to the first meeting time between such a node and $D$ is at most $k \cos(\pi/2 - i \delta)+\delta \leq (k i+1)\delta$. Denote by $\cI$ denote the event that the node $R_f$ belongs to interval $i$. Take now a time interval of length $\frac{1}{(ki+1)\delta}$. Denote by $\cJ$ the event that the first meeting time between node $R_f$ and $D$ is larger than $\frac{1}{(ki+1)\delta}$.  Since the rate parameter of any node belonging to the $i$-th interval and $D$ is at most $1$, we have $\bP[\cJ | \cI] \geq e^{-1}$. Conditioning on $\cJ | \cI$, the meeting time is at least $\frac{1}{(ki+1)\delta}$. Hence we obtain for the total meeting time 
\begin{eqnarray}
\bE[T_2]& \geq& \sum_{i=1}^{2/(\pi \delta)}
\bE[T_2 | \cI]\bP[\cI] \geq\nonumber\\
&\geq&\sum_{i=1}^{2/(\pi \delta)} \bE[T_2 | \cI \wedge \cJ]\bP[\cI] \bP[\cJ | \cI] \geq
\nonumber\\
&\geq& \sum_{i=1}^{2/(\pi \delta)}\left( \frac{2\delta n (1-\epsilon)(1-\eta)(k\cos (i \delta)+\delta)/\pi}{(1+\nu)n(2k/\pi+\delta)}\right.\nonumber\\
&& \left. \frac{1}{(ki+1)\delta}e^{-1}\right).\nonumber
\end{eqnarray}
For $i \leq 2/(100 \delta)$, we have $\cos(i \delta) \geq 1/2$, and the previous sum gives at least $c_0 \sum_{i=1}^{2/(100 \delta)} \frac{1}{i}$ for some $c_0 > 0$. Thus, 
$\bE[T_2] =\Omega(\log(c/\delta))$, for some $c>0$.
 \end{proof}

\begin{theorem}\label{FMIB}
$\bE[T^{ib}_{FM}] \geq \min\{\Omega(\log(1/\delta)), \Omega(n/\log n)\}$.
\end{theorem}
\begin{proof}
As the angle between $S$ and $D$ is at least as large as the angle between $S$ and any node $R_i$, we have that the probability that the first node met by $S$ is different from $D$ is at least $\frac{n}{n+1}$; and that with probability at least $\frac12$, $R_i$ will meet $D$ before $S$ meets $D$. Thus, 
\[
\bE[T^{ib}_{FM}] \geq (1-\frac{1}{n+1})\frac12  \min\{c \log(1/\delta),\frac{c' n}{\log n}\},
\] 
which is $\min\{\Omega(\log(1/\delta)),\Omega(n/\log n) \}$.
\end{proof}

Notice the previous theorem implies that if  $\delta=\delta(n)=o(1)$ then $\bE[T^{ib}_{FM}]\to \infty$.
%Theorem~\ref{IBIB} analyses  IB($\gamma$) routing in presence of interest-based mobility. \remove{Recall that we use $\gamma= \frac{0.29}{m-1}$.  Denote by $N_f$ the random variable counting the number of nodes $R_i$ satisfying $\angle(R_i,S)\leq \arccos \gamma$ and  $\angle(S,R_i)\geq \arccos \frac{3\pi}{8}$. } 
%
\subsection{Interest-based routing}
We now consider the case of $IB(\gamma)$ routing  with interest-based mobility. As before, we set $\gamma=\frac{0.29}{m-1}$ and prove all results for this value of $\gamma$. 
Recall that $N_f$ is the random variable counting the number of  nodes $R_i$ satisfying the conditions of making an angle at most $\arccos \gamma$ with $D$, and at the same time making an angle at most $\arccos \frac{3\pi}{8}$ with $S$.  Define as $\cE$ the event that $N_f \geq 0.99n\frac{1}{4(m-1)}$; by Lemma~\ref{lem:N_f}, this event holds with probability at least $1-e^{-\Theta(n)}$.
\begin{lemma}\label{lem:L1-T1}
Under the conditions stated above, we have $\bE[T_1] = O(1/n)$.
\end{lemma}
\begin{proof}
With event $\cE$ as defined above, we can write $$
\bE[T_1]=\bE[T_1|\cE]\bP[\cE]+\bE[T_1|\hcE]\bP[\hcE].
$$
If $\cE$ holds, by Fact \ref{fact:sum_poisson} the rate parameter of the exponential random variable corresponding to the first meeting time between $S$ and the $N_f$ nodes satisfying the forwarding condition is at least $cn(k \arccos(\frac{3\pi}{8})+\delta)$, for a constant $c > 0$. Thus, $\bE[T_1|\cE]=O(1/n)$, and $\bE[T_1|\cE]\bP[\cE]=O(1/n)$. On the other hand, if $\cE$ does not hold, then  $\bE[T_1|\hcE] \leq \left(n+\frac{1}{\delta n}\right)$, since after time $n$ the first node meeting $S$ is chosen. As $\bP[\hcE]\le e^{-\Theta(n)}$, the contribution of $\bE[T_1|\cE]$ is the dominating one and the statement of the lemma follows.
\end{proof}
\begin{lemma}\label{lem:L2-T2}
Under the conditions above, we have $\bE[T_2] \leq c \gamma/m$ for some constant $c > 0$.
\end{lemma}
\begin{proof}
If $\cE$ holds, denote by $\cF$ the event that in a time interval of length $n$ at least one of the $N_f$ nodes satisfying the forwarding condition meets $S$.  Hence
\begin{eqnarray}
\bE[T_2] &= &\bE[T_2|\cE \wedge \cF] \bP[\cE \wedge \cF]+\nonumber\\
&+& \bE[T_2|\cE \wedge \hcF]\bP[\cE \wedge \hcF]+\bE[T_2|\hcE] \bP[\hcE]\nonumber
\end{eqnarray}
Observe that $\bP[\hcF]\leq e^{-n^2}$. As shown before, $\bP[\hcE] \leq e^{-n}$. If both $\cE$ and $\cF$ hold, then by construction the angle between the node $R_i$ chosen in the first step and $D$ is at most $\arccos \frac{\gamma}{m-1}$, and thus the rate parameter of the first meeting time between $R_i$ and $D$ is at least $k\frac{\gamma}{m-1}$. Therefore, $\bE[T_2|\cE \wedge \cF] \leq \frac{m-1}{\gamma k} =\Theta(1)$. Since in all cases $\bE[T_2]$ can be bounded from above by $\frac{1}{2\delta}$, the case where both $\cE$ and $\cF$ hold is the dominating contribution and the statement follows. 
\end{proof}
Lemma~\ref{lem:L1-T1} and~\ref{lem:L2-T2}  imply the  following theorem:
\begin{theorem}\label{IBIB} %4
For some constant $c > 0$ and any $0<\gamma < 1$, we have
$\bE[T^{ib}_{IB(\gamma)}]\leq m \gamma/c.$
\end{theorem}

Theorems~\ref{FMIB} and ~\ref{IBIB} formally establish the asymptotic superiority of IB($\gamma$) over FM routing in case of interest-based mobility, which is in accordance with intuition.
\section{Simulations}\label{sims}
We have qualitatively verified our asymptotic analysis through simulations, based on both a real world trace collected at the  Infocom 2006 conference -- the trace used in \cite{Mei2,Noulas} --, and the SWIM mobility model of \cite{Mei}, which is shown to closely resemble fundamental features of human mobility. 

\subsection{Real-world trace based evaluation}

A major difficulty in using real-world traces to validate our theoretical results is that   no information about user interests is available, for the vast majority of available traces, making it impossible to realize IB routing. One  exception is the Infocom 06 trace \cite{Hui}, which has been collected during the Infocom 2006 conference. This data trace contains, together with contact logs, a set of user profiles containing information such as nationality, residence, affiliation, spoken languages  etc. Details on the data trace are summarized in Table \ref{infoTrace}.

Similarly to \cite{Mei2}, we have generated 0/1 interest profiles for each user based on the corresponding user profile. Considering that data have been collected in a conference site, we have removed very short contacts (less than $5 min$) from the trace, in order to filter out occasional contacts -- which are likely to be several orders of magnitude more frequent than what we can expect in a non-conference scenario. Note that, according to \cite{Mei2}, the correlation between meeting frequency of a node pair and similarity of the respective interest profiles in the resulting data trace (containing 53 nodes overall) is 0.57. Thus, the Infocom 06 trace, once properly filtered, can be considered as an instance of interest-based mobility, where we expect IB routing to be superior to FM routing.

In order to validate this claim, we have implemented both FM and IB routing. We recall that in case of FM routing, the source delivers the second copy of its message to the first encountered node, while with IB routing the second copy of the message is delivered by the source to the first node whose interest similarity with respect to the destination node is at least $\gamma$. The value of $\gamma$ has been set to $0.29/(m-1)$ as suggested in the analysis, corresponding to 0.0019 in the Infocom 06 trace. Although this value of the forwarding threshold is low, it is nevertheless sufficient to ensure a better performance of IB vs. FM routing.

%\begin{figure}
%\centerline{\includegraphics[height=6.5cm]{segments}}
%\caption{The {\em pdf} of a RWP mobile node can be characterized by computing the expected length of the segment $L_{xy\delta}$ representing the intersection between a random trajectory and square $Q_\delta$ of side $\delta$ centered at $(x,y)$ (shaded area).} \label{segments}
%\end{figure}

The results obtained simulating sending 5000 messages between randomly chosen source/destination pairs are reported in Figure \ref{traceResults}. For each  pair, the message is sent with both FM and IB routing, and the corresponding packet delivery times are recorded. Experiments have been repeated using different TTL (TimeToLive) values of the generated message. Figure \ref{traceResults} reports the difference between the average delivery time with FM and IB routing, and shows that a lower average delivery time is consistently observed with IB routing, thus qualitatively confirming the theoretical results derived in the previous section.

\subsection{Synthetic data simulation}

The real-world trace based evaluation presented in the previous section is based on a limited number of nodes (53), and thus it cannot be used to validate FM and IB scaling behavior. For this purpose, we have performed simulations using the SWIM mobility model \cite{Mei}, which has been shown to be able to generate synthetic contact traces whose features very well match those observed in real-world traces. Similarly to \cite{Mei2}, the mobility model has been modified to account for different degrees of correlation between meeting rates and interest-similarity. We recall that the SWIM model is based on a notion of ``home location" assigned to each node, where node movements are designed so as to resemble a ``distance from home" vs. ``location popularity" tradeoff. Basically, the idea is that nodes tend to move more often towards nearby locations, unless a far off location is very popular. The ``distance from home" vs. ``location popularity" tradeoff is tuned in SWIM through a parameter, called $\alpha$, which essentially gives different weights to the distance and popularity metric when computing the probability distribution used to choose the next destination of a movement. It has been observed in \cite{Mei} that giving preference to the ``distance from home" component of the movement results in highly realistic traces, indicating that users in reality tend to move close to their ``home location". This observation can be used to extend SWIM in such a way that different degrees of interest-based mobility can be simulated. In particular, if the mapping between nodes and their home location is random (as in the standard SWIM model), we expect to observe a low correlation between similarity of user interests and their meeting rates, corresponding to a social-oblivious mobility model. On the other hand, if the mapping between nodes and home location is done based on their interests, we expect to observe a high correlation between similarity of user interests and their meeting rates, corresponding to an interest-based mobility model. 
%
%\begin{figure}[t]
%\centerline{\includegraphics[width=6.5cm]{SWIMResults}}
%\caption{Difference between average packet delivery delay with FM and IB routing with SWIM mobility in the NIM and IM scenario, as a function of the message TTL.
%\label{SWIMResults}}
%\end{figure}

%\acom{ figure 3 is  not readable in black and white; Julinda can you  find a compromise for space and readbility?}

Interest profiles have been generated considering four possible interests ($m=4$), with values chosen uniformly at random in $[0,1]$. In case of interest-based mobility, the mapping between a node interest profile and its ``home location" has been realized by taking as coordinates of the ``home location" the first two coordinates of the interest profile. In  the following we present simulation results referring to scenarios where correlation between meeting rate and similarity of interest profiles is -0.009 (denoted Non-Interest based Mobility -- NIM -- in the following), and 0.61 (denoted Interest-based Mobility -- IM --  in the following), respectively. We have considered networks of size 1000 and 2000 nodes in both scenarios, and sent $10^5$ messages between random source/destination pairs. The results are averaged over the successfully delivered messages. In the discussion below we focus only on average delay. However, we want to stress that in both IM and NIM scenarios, the IB routing slightly outperforms FM in terms of delivery rate (number of messages delivered to destination within TTL): The difference of delivery rates is about 0.015\% in favor of IB. 

Figure~\ref{fig:IM} depicts the performance of the protocols for various values of $\gamma$ on IM mobility. As can be noticed by the figure, the larger the relay threshold $\gamma$, the more IB outperforms FM. Moreover, as predicted by the analysis, the performance improvement of IB over FM routing becomes larger for larger networks. Indeed, for $\gamma = .9$ and $TTL = 24h$, message delivery with IB is respectively $80min$ and $90min$ faster on the network of respectively 1000 nodes (see Figure~\ref{fig:IM-1000}) and 2000 nodes (see Figure~\ref{fig:IM-1000}). This means that, with IM mobility, IB routing delivers more messages with respect to FM, and more quickly.

Notice that the results reported in Figure~\ref{fig:IM} apparently are in contradiction with Theorem ~\ref{IBIB}, which states an upper bound on the expected delivery time which is directly proportional to $\gamma$ -- i.e., higher values of $\gamma$ implies a looser upper bound. Instead, results reported in Figure~\ref{fig:IM} show an increasingly better performance of IB vs. FM routing as $\gamma$ increases. However, we notice that the bound reported in Theorem~\ref{IBIB} is a bound on the {\em absolute} performance of IB routing, while those reported in Figure~\ref{fig:IM} are results referring to the {\em relative} performance of IB vs. FM routing.
 %The results obtained for different TTL values are reported in Figure \ref{SWIMResults}. In both scenarios, the value of $\gamma$ was set to 0.35. 

%\acom{I think that the follwoing should be better motivated; Paolo and Julinda can you take care of it?}
The performance of the protocols with NIM mobility is depicted in Figure~\ref{fig:NIM}. In this case, the performances of the two protocols are very close to each other -- independently of $\gamma$ --, and they become virtually indistinguishable for larger networks. 
The negative values in the figure are due to the few more messages that IB delivers to destination whereas FM does not. Some of these messages reach the destination slightly before the TTL, thus increasing the average delay. However, independently of $\gamma$, the values are close to zero. 
This indicates that, if mobility is not correlated to interest similarity, as far as the average delay is concerned the selection of the relay node is not important: A node meeting the forwarding criteria in IB routing is encountered on average soon after the first node met by the source.

%As seen from the figure, simulation results confirm the trends characterized in the analysis presented in the previous sections: with NIM mobility, FM routing tends to perform better than IB routing (negative values in the plot), especially for networks composed of 1000 nodes. The situation is reversed in the IM scenario: in this case, IB routing outperforms FM routing (positive values in the plot). As predicted by the analysis, the performance improvement of IB over FM routing becomes relatively larger for relatively larger networks.

\subsection{Discussion}
The Infocom 06 trace is characterized by a moderate correlation between meeting frequency and similarity of interest profiles -- the Pearson correlation index is 0.57. However, it is composed of only 53 nodes. Despite the small network size, our simulations have shown that IB routing indeed provides a shorter average message delivery time with respect to FM routing, although the relative improvement is almost negligible (of the order of 0.06\%).

To investigate relative FM and IB performance for larger networks, we used SWIM, and simulated both social-oblivious and interest-based mobility scenarios. Once again, the trend of the results qualitatively confirmed the asymptotic analysis: in case of social-oblivious mobility (correlation index is -0.009), the performance of FM and IB routing is virtually indistinguishable for all network sizes; on the other hand, with interest-based mobility (correlation index is 0.61), IB routing provides better performance than FM. It is interesting to observe the trend of performance improvement with increasing network size: performance is improved by about 
5.5\% for 1000 nodes, and by about 6.25\% for 2000 nodes. Although percentage improvements over FM routing are modest, the trend of improvement is clearly increasing with network size, thus confirming the asymptotic analysis. Also, IB forwarding performance improvement over FM forwarding becomes more and more noticeable as the value of $\gamma$, which determines selectivity in forwarding the message, becomes higher: with $\gamma=0.2$ and 2000 nodes, IB improves delivery delay w.r.t. FM forwarding of about 0.1\%; with $\gamma = 0.6$ improvement becomes 1.7\%, and it raises up to 6.25\% when $\gamma=0.9$.

\section{More copies and more hops}\label{MoreHops}

In this section, we extend the analysis of sections \ref{WorstCase} and \ref{IBanalysis} under several respects. To start with, we consider a variation of the FM routing protocol for the case of interest-based mobility, which we call FM*. In this variation, we assume that the message is forwarded between two nodes only if the new node (its interest profile) is {\em closer} (i.e., more similar) to the destination than the node currently carrying the message. Also, we assume that if a node has already forwarded the message to a set of nodes, then it will forward the message only to nodes which are closer to the destination than all the previous ones. Clearly, FM* performs better than FM routing in presence of interest-based mobility, since it at least partially accounts for similarity of interest profiles when forwarding messages. Note that the difference between FM* and IB routing is that, while in the latter a minimum similarity threshold between potential forwarders and destination must be met, in the former even a tiny improvement of similarity w.r.t. destination of the potential forwarder with respect to current forwarders is enough to forward the message. In this respect, FM* somewhat resembles delegation forwarding \cite{Erram}.

The routing protocols considered in this section are extensions of FM and IB under the following respect. The source node $S$ initially carries an arbitrary number $q\ge 2$ of message copies (and not just 2 copies as in the original protocols). Furthermore, if a node $A$ currently carrying $k$ copies of the message meets a new forwarding node $B$, it will deliver to $B$ exactly $\lfloor k/2\rfloor$ copies of the message, keeping the remaining $\lceil k/2\rceil$ for itself. When a node is left with a single copy of the message, it can deliver this copy only to the destination. Notice that, by setting $q$ to an arbitrary power of 2, the extended version of FM is equivalent to BinarySW \cite{Spyro}. Notice also that, if $q >2$, the extended versions of the routing protocols allow delivery of messages from $S$ to $D$ along paths of hop-count larger than 2.

In the next subsection, we consider a version of FM* where multi-hop propagation of a message from $S$ to $D$ is allowed. In other words, if $A$ and $B$ are the two nodes currently carrying a copy of $M$ -- we retain the assumption of at most two message copies circulating in the network --, either of them -- say, $A$ --, can deliver its copy to another node $C$ if $C$'s interests are more similar to $D$'s than those of node $A$. This process is repeatable, up to a maximum length of $\ell$ in the message propagation path ($\ell=2$ in the original protocols).

First, we observe that, for any mobility model and any routing algorithm, it is clear that the expected meeting times of $\ell>2$ hops and $q>2$ copies are always at most as large as the expected meeting times of the case of $2$ copies and $2$ hops. Thus, upper bounds on the asymptotic performance provided by IB routing remains valid also for $\ell,k>2$. We now show that, even by allowing more copies and/or hops and a smarter forwarding strategy (the FM* approach), the expected meeting time of FM routing in both mobility models does not improve asymptotically. 

For presentation purposes, in the following we will exploit the well-known relation between Poisson point processes and exponentially distributed r.v.s, namely the fact that the time for the first hit in a Poisson point process of intensity $\mu$ is an exponentially distributed r.v. of rate parameter $\mu$. Thus, by ``intensity of the Poisson process between A and B" we mean ``the rate of the exponentially distributed r.v. corresponding to the first meeting time between A and B". In order to simplify the presentation of the statements, by the observation made in the beginning of the proof of Lemma~\ref{lem:T2}, we will assume that $\delta=\omega(\log n/n)$.
\subsection{$\ell$ hops}\label{sub-sec-hops}
 First we consider the case of FM* routing  with $\ell \geq 2$ ($\ell$ constant) hops and $2$ copies only. We denote by $T_1$ the random variable counting the time it takes for $S$ to meet the first node out of $\{R_1,\ldots,R_n,D\}$. Denote by $R_{r(i)}$ the $i$-th node met by $S$,  and, for $i=2,\ldots,\ell-1$,  let  $T_i$ be the random variable counting the time it takes for $R_{r(i-1)}$ to meet $R_{r(i)}$ (we assume that if $D$ was met already in previous steps, then $T_i=0$). $T_{\ell}$ finally is the random variable counting the time it takes for the first out of $\{S,R_{r(\ell-1)}\}$ to meet $D$ (if $D$ was met in previous rounds then $T_{\ell}=0$). 

 In the case of social-oblivious mobility, we have $\bE[T_1]=\frac{1}{\lambda(n+1)}$, $\bE[T_i]\le \frac{1}{\lambda n}$ for $i=2,\ldots,\ell-1$ and $\bE[T_{\ell}] \le \frac{1}{2\lambda}$. By a similar discussion as in Section~\ref{WorstCase}, $\bE[T^{so}_{FM^*}] \le \frac{1}{2\lambda}(1+o(1))$. However it is possible to show that the probability that r.v. $T_i$, $i=1,2, \ldots, \ell$ is zero is negligible and, hence, we are able to state  that
 $\bE[T^{so}_{FM^*}] = \frac{1}{2\lambda}(1+o(1))$. 
 
 In the case of interest-based mobility, we first need the following lemma:
 \begin{lemma}\label{lem:constantalpha}
 There exist constants $\alpha > 0$  and $\beta  > 0$ such that, with probability at least $\alpha$, the first $\ell-1$ nodes that serve as intermediate hops all make up an angle of at most $\frac{\ell-1}{\beta \delta}$ with $S$.
 \end{lemma}
\begin{proof}
See Appendix.
\end{proof}

The following is an immediate consequence of the previous lemma.
 \begin{corollary}\label{cor:needkrounds}
 With probability at least $\alpha > 0$, $D$ is not yet found among the $\ell-1$ vertices that serve as intermediate hops.
 \end{corollary}
The following lemma extends Lemma \ref{lem:T2} to the case of $\ell$ hops.
 \begin{lemma}\label{lem:T2_khops}
 Under the assumptions above, we have $\bE[T_{\ell}] \ge c \log(1/\delta)$ for some positive constant $c$.
 \end{lemma}
 \begin{proof}
See Appendix.
 \end{proof}

 By combining Corollary~\ref{cor:needkrounds} and Lemma~\ref{lem:T2_khops}, we have proved the following:
\begin{theorem}
 $\bE[T_{FM^*}^{ib}] = \Omega(\log(1/\delta)).$
\end{theorem}
\subsection{Using $q$ copies and $\ell$ hops}
We now discuss how to extend the model of $\ell$ hops to the model where $q \geq 2$ copies of a message are used. We assume without loss of generality that $q=2^w$ for some natural number $w$.
%
%, and we assume that the copies of the message are forwarded between the hops in the following way: whenever a node contains $2^s$, $s \geq 1$  copies of a message and meets another node different from $D$, in the independent mobility model it always  gives to that node $2^{s-1}$ copies of that message. In the interest-based mobility model it gives to that node $2^{s-1}$ copies only if that new node is closer to $D$ than all previous hops containing some copies of the message. 
%We also assume that all vertices keep the last copy for itself and deliver it only if $D$ is met; therefore the number of hops is bounded by  $\log_2 q$. 

We start with the following straightforward observation.
\begin{observation}\label{lem:help}
The number of relay nodes (excluding $S$) is at most $\ell-1$, and exactly $\ell-1$ if $D$ is not among them.
\end{observation}

We define the Poisson point process between two vertices $U$ and $V$, $U,V \in \{S,R_1,\ldots,R_n,D\}$ as \emph{active} at time $t$ if at time $t$ node $U$ has more than one copy of the message, $V$ does not yet have a copy, and $V$ is closer to $D$ than all vertices containing already copies of messages. Define by $T_1$ the random variable counting the time it takes for $S$ to meet the first out of $\{R_1,\ldots,R_n,D\}$. $T_i$, $i=2,\ldots,q-1$  is the random variable counting the time of the first meeting of all active Poisson point processes at time $T_1+\ldots+T_{i-1}$ from time $T_1+\ldots+T_{i-1}$ onwards ($T_i = 0$ if $D$ has been met before).  $T_q$ is the random variable counting the time of the first meeting of all Poisson processes between vertices that have one copy of the message at time $T_1+\ldots+T_{q-1}$ and $D$ ($T_q = 0$ if $D$ has been met before). 

Observe that for the FM* routing algorithm in the social-oblivious mobility model, we have $\bE[T_1]=\frac{1}{\lambda(n+1)}$, $\bE[T_i] \leq \frac{1}{\lambda(n+1)}$ for $i=1,\ldots,q-1$, and $\bE[T_q] \le \frac{1}{q \lambda}$. By the same argument as in Subsection~\ref{sub-sec-hops}, we can show that $\bE[T_{FM^*}^{so}]=\frac{1}{q \lambda}(1+o(1))$. Thus, in this model the expected message delivery time for $q > 2$ is a smaller constant. 

Now we consider FM* routing in the interest-based mobility model. Call $R_{r(1)},\ldots,R_{r(h-1)}$ the intermediate nodes in order of their appearance, i.e., $R_{r(i)}$ contains at least one copy of the message from time $T_1+\ldots+T_i$ on, for any $i=1,\ldots,h-1$. 

\begin{lemma}\label{lem:constantalpha2}
With probability at least $\alpha > 0$, $D$ is not among the first $h-1$ hops containing at least one copy of the message.
\end{lemma}
\begin{proof}
See Appendix.
\end{proof}

The previous lemma states that, with probability at least $\alpha$, all of the $R_{r(1)},\ldots,R_{r(q-1)}$ nodes make an angle of at most $\frac{q-1}{\beta \delta}$ with $S$. 

We are now ready to state the main result of this section.
\begin{theorem}\label{ThmMultiCopies}
Assume $S$ has $q$ copies of $M$ and we can make up to $\ell = \log_2 q$ hops, then 
$$\bE[T_{FM^*}^{ib}]=\Omega(\log(1/\delta)).$$
\end{theorem}
\begin{proof}
See Appendix.
 \end{proof}

\subsection{Unknown destination}
  A major limitation of IB routing is that the sender is assumed to know the interest profile of the destination, i.e., the coordinates $D[a_1,a_2,\ldots ,a_m]$ of $D$ in the interest space.  We now relax this assumption assuming  that $S$ knows the identity of node $D$ (so delivery of $M$ to $D$ is possible), but not its interest profile, and we  show that a modified version of the $IB(\gamma)$ routing that uses more than one copy of the message  also provides asymptotically the same upper bound as the original version of IB routing. 

The idea is that the routing protocol chooses $m-1$ relay nodes (i.e., the number of message copies equals the number of dimensions in the interest space) with the characteristic that each one the $m-1$  relay nodes will be ``almost orthogonal" to the others and to  $S$, and $S$ will pass a copy  to each one of them, and keep one. 
Therefore, when  $S$ decides whether or not to forward a copy of $M$ to a possible  $R_i$, $S$ has only information of $R_i[a_1,a_2,\ldots ,a_m]$. 
Let $\hat{R}_j$ denote the $j$-th relay chosen node, $j=1,2,\ldots, m-1$.  
We consider the following   routing algorithm \texttt{Mod-}$IB(\gamma)$ to choose relay nodes:
If $S$ meets a node 
with coordinates $R_i[r_1,r_2, \ldots, r_m]$, the node becomes the $j$-th relay node $\hat{R}_j$, $j=1,2, \ldots, q-1$, if the following  conditions are met:
\begin{itemize}
\item[--] $0.05 \leq R_i[1] \leq 0.1$;  
\item[--] $\exists k,\, 2 \leq k \leq m$ s. t. $0.8 \leq R_i[k] \leq 0.85$;  
\item[--] $\forall s,\, 1,\leq s\leq j-1$, $\hat{R}_s[k] < 0.8$.
\end{itemize}

\begin{theorem}\label{thm:undest}
For a constant $c>0$ and  $\gamma=\frac{0.29}{m-1}$ we have  $\bE[T_{\mathtt{Mod-}IB(\gamma)}^{ib}]\leq m \gamma/c$.
\end{theorem}
\begin{proof}
See Appendix.
\end{proof}
\section{Uniform distribution of the destination $D$}\label{Unif}
So far, we have  considered $S$ and $D$ to have orthogonal coordinates in the interest space.
In this section, we extend the analysis to the case where the source keeps its coordinates $S[1,0,\ldots, 0]$, but we choose $\angle (S,D)$ uniformly at random in $[0,\pi/2]$. We show that under this average-case assumptions, the original FM routing algorithm takes constant time also with interest-based mobility, i.e., it has the same asymptotical performance as IB routing.

\begin{theorem}\label{thm:averagecase}
Assume the angle between $S$ and $D$ is chosen uniformly at random in $[0,\pi/2]$. Then, $\bE[T^{ib}_{FM}]=O(1).$
\end{theorem}
\begin{proof} 
See Appendix.
\end{proof}
Note  that a routing algorithm without intermediate hops, call it $FM^0$, according to which $S$ can only directly deliver the message to $D$, still needs more than constant time in expectation.
\begin{lemma}\label{lem:0hops}
Under the above assumptions, we have $\bE[T_{FM^0}]=\Omega(\log(1/\delta)).$
\end{lemma}
\begin{proof}
See Appendix.
 \end{proof}
 
Comparing Theorem~\ref{thm:averagecase} with Theorem~\ref{IBIB}, and observing that average case performance of IB routing can be no worse than its performance in the worst case, we can conclude that FM and IB routing yield the same asymptotic performance in the average case.

 \section{Conclusion}\label{conclusions}

 We have formally analyzed  and experimentally validated the  delivery time under mobility and forwarding scenarios accounting for social relationships between network nodes. The main contribution of this paper is proving that, under fair conditions for what concerns storage resources, social-aware forwarding is asymptotically superior to social-oblivious forwarding in presence of interest-based mobility: its performance is never below, while it is asymptotically superior under some circumstances, namely, orthogonal interests between sender and destination. 

As a byproduct, our analysis provides interesting insights on the design of social-aware forwarding strategies; for instance, our results indicate that when the interest profile of the destination is not known to the source node, a good strategy is trying to deliver a copy of the message to forwarding nodes with ``almost orthogonal" interests, in order to increase the chances that at least one of them is near to the destination in the interest space and, hence, likely to meet the destination soon according to the interest-based mobility model.
 
We believe several avenues for further research are disclosed by our initial results, such as considering scenarios in which individual interests evolve in a short time scale, or scenarios in which forwarding of messages is probabilistic instead of deterministic.

\section{Acknowledgements}

The three first authors were partially supported by the EU through project FRONTS. The work of P. Santi was partially supported by MIUR, program PRIN, Project COGENT. The initial ideas underlying this work were developed when some of the authors were visiting the Centre de Recerca Matematica, Barcelona, Spain.

\newpage

\section{Appendix}

{\bf Proof of Lemma \ref{lem:constantalpha}.} As in Section \ref{IBanalysis}, we partition the interval $[0,\pi/2]$ into subintervals of length $\delta$ and we extend  Lemma~\ref{lem:T2}  by showing that the probability that node $R_{r(1)}$ is chosen from the $i$-th subinterval is at least
 $$\frac{2\delta n (1-\eta)(1-\epsilon)(k\cos (i \delta)+\delta)/\pi}{(1+\nu)n(2k/\pi+\delta)}.$$ 
 For $i \leq \frac{1}{100\delta}$, this probability is at least $\frac{\delta}{c}$ for some constant $c > 0$. Choose $\beta=\beta(\ell)$ to be a sufficiently large constant. Thus, the probability that $R_{r(1)}$ is in the first $\frac{1}{\beta \delta}$ subintervals, is at least $\eta_1$ for some $\eta_1 > 0$. By conditioning under this event, by a similar argument as in Lemma~\ref{lem:T2}, we can prove using Theorem A.1.15 of~\cite{AlonSpencer}, that, with high probability, the intensity of the Poisson point processes between $S$ (or the chosen node $R_{r(1)}$ in the first $\frac{1}{\beta \delta}$ subintervals) and all nodes whose angle w.r.t. the destination is to the right of these subintervals is at least $\Omega(n)$. 
 
 We now recall that all vertices closer to $D$ than  $R_{r(1)}$ are possible next hops. Using this observation, we now iterate the previous reasoning: with probability at least $\eta_2$ the second node $R_{r(2)}$ is among the first $\frac{2}{\beta \delta}$ subintervals (and not among the first $\frac{1}{\beta \delta}$ subintervals), and in general with probability $\eta_i$ the node $R_{r(i)}$ is among the first $\frac{i}{\beta \delta}$ subintervals, for any $i=1,\ldots,\ell-1$. Therefore, with probability at least $\alpha := \prod_{i=1}^{\ell-1} \eta_i$, the first $\ell-1$ intermediate nodes form angles of at most $\frac{\ell-1}{\beta \delta}$ with $S$, thus proving the lemma. 
 \bigskip
 
{\bf Proof of Lemma \ref{lem:T2_khops}.} To prove the lemma, we first use  Lemma~\ref{lem:constantalpha}, that implies that nodes $R_1,\ldots,R_{\ell-1}$ all make an angle of at most $\frac{\ell-1}{\beta \delta}$ with probability at least $\alpha$. 
 Thus, conditioning under this event, we apply a similar argument as in Lemma~\ref{lem:T2}: denoting by $E_j$ the event that node $R_j$ is in subinterval $i_j$, and denoting by $E_{1(\ell-1)}$ the event that $E_1 \wedge \ldots \wedge E_{\ell-1}$, we have that there exist constants $c_0,c_1 > 0$ such that
 \begin{eqnarray}
 \bE[T_{\ell}] & \geq &
  \sum_{i_1=1}^{\frac{\ell-1}{\beta \delta}}\ldots \sum_{i_{\ell-1}=i_{\ell-2}}^{\frac{\ell-1}{\beta \delta}} \bE[T_{\ell}|E_{1(\ell-1)}]  \bP[E_{1(\ell-1)}] \nonumber \\
& =&\sum_{i_1=1}^{\frac{\ell-1}{\beta \delta}}\ldots \sum_{i_{\ell-1}=i_{\ell-2}}^{\frac{\ell-1}{\beta \delta}} \bE[T_{\ell}|E_{1(\ell-1)}] \nonumber \\
  && \bP[E_1]\bP[E_2 | E_1] \ldots \bP[E_{\ell-1} | E_1 \wedge \ldots \wedge E_{\ell-2}] \nonumber \\
 &\geq &\sum_{i_1=1}^{\frac{\ell-1}{\beta \delta}}\ldots \sum_{i_{\ell-1}=i_{\ell-2}}^{\frac{\ell-1}{\beta \delta}} \frac{1}{(ki_{\ell-1}+1)\delta}e^{-1} (c_0 \delta)^{\ell-1} \nonumber \\
 &\geq & c_1 e^{-1} \sum_{i_{\ell-1}=i_{\ell-2}}^{\frac{\ell-1}{\beta \delta}} \frac{1}{(ki_{\ell-1}+1)}, \nonumber 
 \end{eqnarray}
 and therefore $\bE[T_{\ell}] \ge c \log(1/\delta))$ for some positive constant $c$.
 
 \bigskip
 
 {\bf Proof of Lemma \ref{lem:constantalpha2}.} We first observe that Lemma~\ref{lem:constantalpha} also applies in the case when we consider Poisson point processes between any node out of $\{S,R_{r(1)},\ldots,R_{r(i)}\}$ (chosen from the first $\frac{i}{\beta \delta}$ subintervals) and a node to the right of $R_{r(i)}$: no matter which node is chosen out of $\{S,R_{r(1)},\ldots,R_{r(i)}\}$, the probability of choosing one node from the subinterval following $R_{r(i)}$ is still $\Theta(1/\delta)$, since the total intensity of all Poisson point processes between $R_{r(i)}$ and the vertices to the right of $R_{r(i)}$ is still $\Theta(n)$. Thus, by considering the $\frac{1}{\beta \delta}$ subintervals following $R_{r(i)}$, we can show  that, with constant probability, the next node  $R_{r(i+1)}$ belongs to these subintervals. 

By multiplying all constants of all $q-1$ steps, we can show  that, with probability at least $\alpha$, all $q-1$ intermediate nodes  form an angle of at most $\frac{q-1}{\beta \delta}$ with any node out of $\{S,R_{r(1)},\ldots,R_{r(i)}\}$, and thus with at least that probability $D$ is not among these $q-1$ nodes.
 
\bigskip

{\bf Proof of Theorem \ref{ThmMultiCopies}.} To prove the result, we  have to show that $\bE[T_q]=\Omega(1/\delta).$
Observe that if the Poisson point process between $R_{r(q-1)}$ and $D$ has at most certain intensity $\mu$, all other Poisson point processes between $\{R_{r(i)}\}_{i=1}^{q-2}\cup \{S\}$ also have intensity at most $\mu$. Moreover, these Poisson point processes are independent, and their superposition gives rise, by Fact~\ref{fact:sum_poisson}, to a new Poisson point process with intensity at most $q \mu$. Thus, using similar arguments as in the proof of Lemma~\ref{lem:T2_khops}, we can split the value of $\bE[T_q]$ according to the subintervals of length $\delta$ to which node $R_{r(i)}$ belongs  (assuming that all of them are among the first $\frac{q-1}{\beta \delta}$ subintervals),  obtaining 
\begin{eqnarray}
 \bE[T_{q}] & \geq & c_1 e^{-1} \sum_{i_{q-1}=i_{q-2}}^{\frac{q-1}{\beta \delta}} \frac{1}{q(ki_{q-1}+1)}, \nonumber
 \end{eqnarray}
 and thus $\bE[T_q] \geq \frac{c_2 \log(1/\delta)}{q}$. Since $q$ is assumed to be constant, $\bE[T_q]=\Omega(1/\delta).$
 
 \bigskip
 
 {\bf Proof of Theorem \ref{thm:undest}.} Assume that $\hat{R}_j$ is only accepted as $j$-th relay node if the value in the $j$-th coordinate is between $0.8$ and $0.85$.
First we will show that $\bE[T_1]=O(1/n)$. Observe that for any node $R$ (except for $D$), 
\[
\PR{0.05 \leq \cos \angle (S,R) \leq 0.1}  \geq 0.03.\]
Conditioned under making such an angle, the sum of the squares of all other coordinates of $R$ is at least $0.99$, and once the angle is chosen,  the position on the sphere  is selected uniformly at random from all remaining positions, so the position of $v$ is chosen from the surface of an $(m-1)$-dimensional sphere of squared radius $\geq 0.99$. Fix a coordinate $j \geq 2$ in which we would like to have $0.8 \leq R[j] \leq 0.85$. 
The intersection of an $m$-dimensional sphere of squared radius at least $0.99$ centered at the origin, with the region bounded by the two parallel hyperplanes whose values in dimension $j$ are $0.8$ and $0.85$, respectively, has a surface area which is bigger than the one of an $(m-1)$-dimensional sphere of squared radius at least $0.51$. \remove{(this value is attained, if the value in dimension $j$ is exactly $0.85$) multiplied by $0.05$. }Intersecting this area with a hyperplane having some fixed value in the first coordinate between $0.05$ and $0.1$ \remove{(accounting for the initially chosen angle $\angle(S,D)$) }yields a surface area of at least the surface area of an $(m-2)$-dimensional sphere of squared radius $0.5$ times \remove{the length in dimension $j$, which is} $0.05$. Thus, denoting by $S_m(r)$ the surface area of an $m$-dimensional sphere of radius $r$, and denoting by $\cV$ the event that vertex $R$ satisfies the conditions for being selected as relay node $\hat{R}_j$, we obtain 
$$\PR{\cV}  \geq 0.03 \frac{0.05 S_{m-2}(\sqrt{0.5})}{S_{m-1}(1)},$$
where $c>0$ is obtained 
from the  the fact that 
$S_m(r)=\frac{2\pi^{m/2}r^{m-1}}{\Gamma(m/2)}$. Thus, we have an  expected number of potential vertices satisfying the condition to  be selected as  $\hat{R}_j$. Since all vertices are independent, we have that with probability $\geq 1-e^{-\Theta(n)}$ this number is linear. Taking a union bound,
we get that this also holds for all $m$ dimensions with the same probability. Since the intensity between any vertex eligible as $\hat{R}_j$ and $S$ is constant, conditioning under the event $\cal{E}$ of having a linear number of possible relay nodes, $\bE[T_1 |\cal{E}]$ $= O(1/n)$.  Since $\PR{\cal{\bar{E}}} = e^{-\Theta(n)}$,  the dominating contribution comes from $\bE[T_1 | \cal{E}]$, so $\bE[T_1]=O(1/n).$\\
To finish the proof, we show that there exists  $c > 0$ such that
$\bE[T_2] \leq m \gamma/c$: observe that since $D$ is a vector in the positive orthant  of the $m$-dimensional sphere, in at least one dimension its coordinate has to be at least $\frac{1}{\sqrt{m}}$. As there exists some  $\hat{R}_j$ whose value in this coordinate is  $\geq 0.8$, $\cos \angle(\hat{R}_j,D) \geq \frac{0.8}{\sqrt{m}}$. For any $m \geq 2$, $\frac{0.8}{\sqrt{m}} > \frac{\gamma}{m-1}$, and thus, the same analysis as above gives the upper bound of $m \gamma/c$.

\bigskip

{\bf Proof of Theorem \ref{thm:averagecase}.} We restrict ourselves to the $2$-dimensional case. Denote by $R$ the relay node chosen.
Since the first node met by $S$ is selected as relay node, $\bE[T_1]=O(1/n)$. We will give an upper bound on $\bE[T_2]$  following the same ideas as before: we  split the angles between $S$ and $D$ as well as the positions of $R$, into intervals of length $\delta$. Denote by  $\cI$ the event that $R$ is in the $i$-th interval, and denote by $\cY$ the event that $D$ is in the $j$-th interval, for $1 \leq i,j \leq \frac{\pi}{2 \delta}$. Then, 
\[
 \bE[T_2]=\sum_i \sum_j \bE[T_2 | \cI \wedge \cY] \PR{\cI} \PR{\cY}~.
 \] 
 By definition, $\PR{\cY}=2\delta/\pi$, and by  Chernoff bounds, $\PR{\cI} \leq c \delta$ for  small $c > 0$. \\
 Observe that for $j \leq \frac{3\pi}{8 \delta}$, $\cos \angle (S,D) \geq 0.38$, and therefore $\bE[T_2 | \cY] =O(1)$. Moreover, if $i \geq \frac{\pi}{8 \delta}$, $\cos \angle(R_i,D) \geq 0.38$, and thus $\bE[T_2 |\cI] =O(1)$.  Assume  $i < \frac{\pi}{8 \delta}$ and $j > \frac{3\pi}{8 \delta}$. Given  events $\cI$ and $\cY$, $\angle (R,D) \leq j\delta -(i-1)\delta$, and hence $\cos \angle (R,D) \geq \cos (j-(i-1))\delta$. \\
 Therefore, given $\cI$ and $\cY$, the expected time is at most 
 $\frac{1}{k \cos \left( (j-i+1)\delta \right)}$. Thus, $$\bE[T_2] \leq O(1)+c \delta^2 \sum_{1 \leq i < \frac{\pi}{8\delta}} \sum_{j > \frac{3\pi}{8 \delta}} \frac{1}{k \cos \left( (j-i+1)\delta \right)}.$$ Writing $j=\frac{\pi}{2}-t$, we get
 \[ \bE[T_2] \leq c \delta^2 \sum_{1 \leq i < \frac{\pi}{8\delta}} \sum_{0 \leq t<\frac{\pi}{8 \delta}} \frac{1}{k \sin \left( (t+i-1)\delta \right)}~.\] 
 Using the bound
 \[ \sin ((t+i-1)\delta) \geq (t+i-1)\delta -((t+i-1)\delta)^3/6 \geq \frac{5}{6}(t+i-1)\delta~,
 \] we obtain 
 \[
 \bE[T_2] \leq c \delta \frac{6}{5k} \sum_i \sum_t \frac{1}{t+i-1}~.
 \]
 Setting $u=i-1$, we have $\bE[T_2]  \leq c' \delta \sum_{u=0}^{\frac{\pi}{8\delta}} \sum_{t=0}^{\frac{\pi}{8\delta}} \frac{1}{t+u}$ for some  $c' > 0$. \\
 The cases where $i=1$ or $t=0$ or $t=1$ all happen with probability at most $c \delta$, and since the intensity between any pair of points is at least $\delta$,  the contribution of these cases to $\bE[T_2]$ is at most $O(1)$. Hence, we exclude these cases to get 
 \begin{eqnarray}
 \bE[T_2]  &\leq& O(1)+c' \delta \sum_{u=2}^{\frac{\pi}{8\delta}} \sum_{t=2}^{\frac{\pi}{8\delta}} \frac{1}{t+u}\nonumber\\
  &\leq& O(1)+c' \delta
 \int_{u=1}^{\frac{\pi}{8\delta}} \int_{t=1}^{\frac{\pi}{8\delta}} \frac{1}{t+u} \mbox{\emph{dt du}}~.
 \end{eqnarray}
 Since,$\int_{t=1}^{\frac{\pi}{8\delta}} \frac{1}{t+u} dt= \log(\frac{\pi}{8\delta}+u)-\log(1+u)$,  we have
 $$
 \begin{array}{l}
 \int_{u=1}^{\frac{\pi}{8\delta}} \left( \log(\frac{\pi}{8\delta}+u)-\log(1+u)\right) du  \\
= \left((\frac{\pi}{8\delta}+u)\log(\frac{\pi}{8\delta}+u)-(1+u)\log(1+u)\right)|_{u=1}^{\frac{\pi}{8\delta}} \\
=\frac{\pi}{4\delta}\log \frac{\pi}{4\delta}-2(\frac{\pi}{8\delta}+1)\log(\frac{\pi}{8\delta}+1)+2\log2 \\
\leq \frac{\pi}{4\delta}(\log \frac{\pi}{4\delta}-\log(\frac{\pi}{8\delta}+1))+2\log 2 \\
\leq \frac{\pi}{4\delta} \log 2 + 2\log 2.
\end{array}
$$ 
Thus, $\bE[T_2] = O(1)+c'\delta (\frac{\pi}{4\delta}\log 2 +2 \log2)=O(1)$, and the statement of the theorem follows.  

\bigskip

{\bf Proof of Lemma \ref{lem:0hops}.} Denoting by $\alpha=\angle(S,D)$, we have $$\bE[T_{FM^0}]=\int_{0}^{\pi/2} \frac{2}{\pi}\frac{1}{k \cos \alpha + \delta} d\alpha.$$ For $\alpha > (\pi/2 - \delta)$, the intensity of the Poisson process is $\geq\delta$, and since such  value of $\alpha$ is chosen with probability  $\leq\frac{2 \delta}{\pi}$, the total contribution of this case to $\bE[T_{FM^0}]$ is $O(1)$. Hence, consider only $\alpha < (\pi/2 - \delta)$. In this case, for a suitably chosen constant $c>0$, the intensity of the Poisson  process between $S$ and $D$ is at most $c k\cos \alpha$. Thus, 
\begin{eqnarray}
\bE[T_{FM^0}] & \geq  & \int_0^{(\pi/2-\delta)} \frac{2}{\pi} \frac{1}{ck\cos \alpha} d\alpha \nonumber\\
&=&\frac{2}{c \pi k} \left(\log (\sin(\alpha/2)+\cos(\alpha/2))-\log(\cos(\alpha/2)\right.\nonumber\\
&-&\left.\sin(\alpha/2))\right)|_{\alpha=0}^{(\pi/2-\delta)}~.\nonumber
\end{eqnarray}
 Evaluating the integral, we obtain that this term is at least $$ -c'\log \left(\cos((\pi/2-\delta)/2)-\sin((\pi/2-\delta)/2)\right)$$  for some  $c' > 0$. Making a Taylor series expansion for the expression inside the logarithm around the point $\pi/4$, we see that this expression is $\delta \sin(\pi/4)+O(\delta^2)$. 

\newpage
\begin{table}
\begin{center}
\begin{tabular}{|l|c|}
\hline
Experimental data set & Infocom~06\\ \hline
Device  & iMote\\
Network type & Bluetooth\\
Duration (days)& 3\\
Granularity (sec)& 120\\
Participants with profile & 61\\
Internal contacts number& 191,336\\
Average Contacts/pair/day
& 6.7\\[1mm] \hline
\end{tabular}
\caption{Detailed information on the Infocom~06 trace.}
\label{infoTrace}
\end{center}
\end{table}

\begin{figure}[h]
\centerline{\includegraphics[width=5cm]{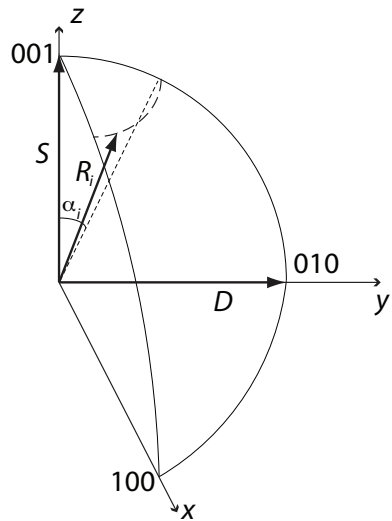}}
\caption{Node $S$ and $D$ in the unit sphere, and random choice of the angles between nodes: first, the angle $\alpha_i$ is chosen uniformly at random in $[0,\pi/2]$; then, a point $R_i$ is chosen uniformly at random in the $(m-1)$-dimensional space obtained by fixing angle $\alpha_i$ w.r.t. node $S$. In this example, we have $m=3$.}
\label{sphere}
\end{figure}

\begin{figure}[t]
\centerline{\includegraphics[width=6.5cm]{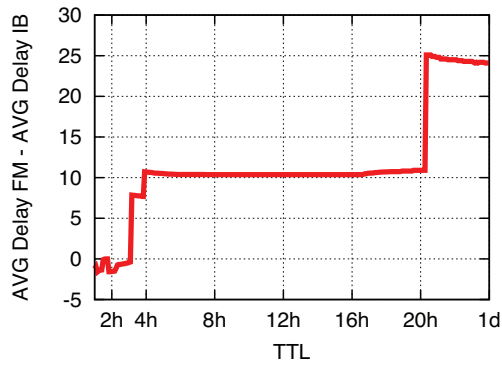}}
\caption{Difference between average packet delivery delay with FM and IB routing with the Infocom 06 trace as a function of the message TTL.
\label{traceResults}}
\end{figure}

\begin{figure}[t]
    \begin{center}
        \subfigure[IM network of 1000 nodes.]
        {
            \label{fig:IM-1000}
            \includegraphics[width=.45\textwidth]{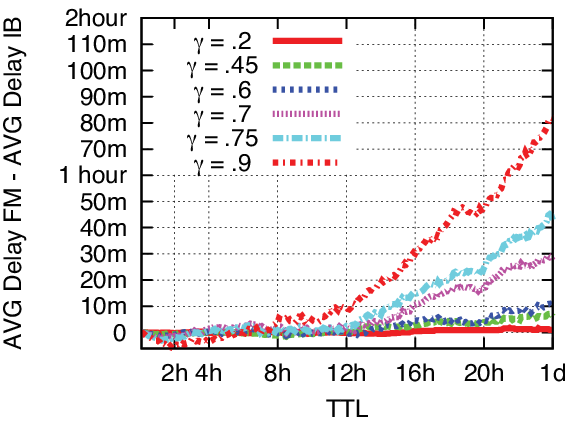}
        }
        \subfigure[IM network of 2000 nodes.]
        {
            \label{fig:IM-2000}
            \includegraphics[width=.45\textwidth]{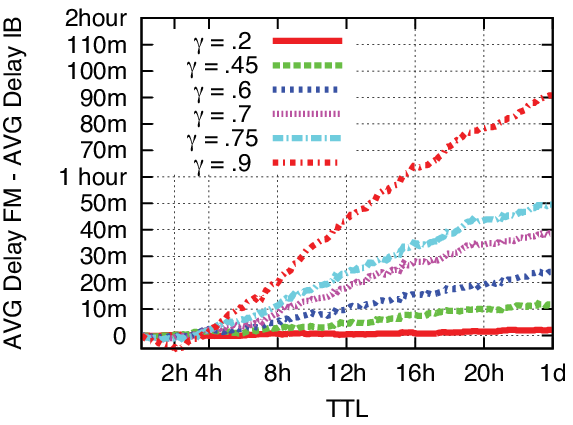}\vspace{-7mm}
        }
		\caption{Difference between average packet delivery delay with FM and IB routing with SWIM mobility in the Interest-based mobility (IM) scenario, as a function of the message TTL.}\vspace{-7mm}
        \label{fig:IM} 
    \end{center}
\end{figure}

\begin{figure}
    \begin{center}
        \subfigure[NIM network of 1000 nodes.]
        {
            \label{fig:NIM-1000}
            \includegraphics[width=.45\textwidth]{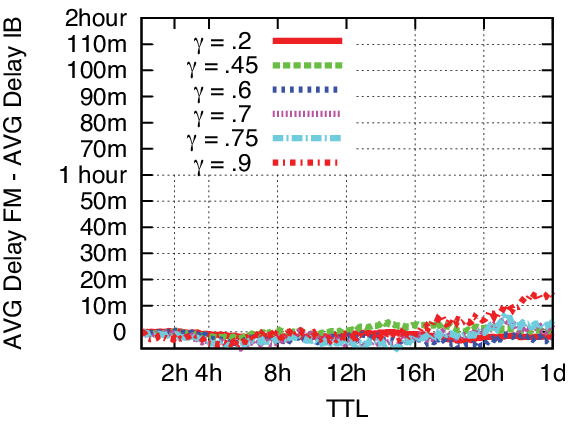}
        }
        \subfigure[NIM network of 2000 nodes.]
        {
            \label{fig:NIM-2000}
            \includegraphics[width=.45\textwidth]{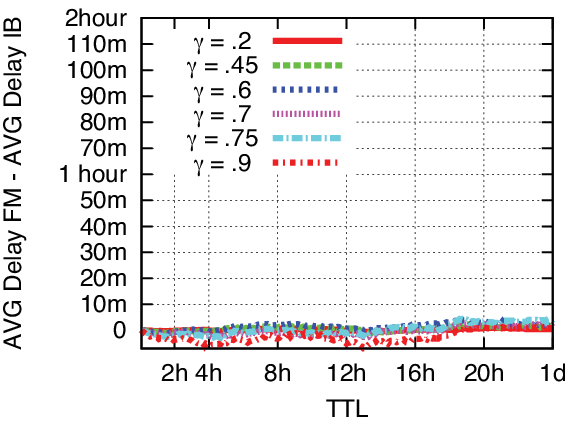}\vspace{-7mm}
        }
		\caption{Difference between average packet delivery delay with FM and IB routing with SWIM mobility in the Non Interest-based mobility (NIM) scenario, as a function of the message TTL. 		}\vspace{-11mm}
        \label{fig:NIM}
    \end{center}
\end{figure}


\begin{thebibliography}{01}

\bibitem{AlH1}A.~Al Hanbali, P.~Nain, E.~Altman, ``Performance of Ad Hoc Networks with two-hop relay routing and limited packet lifetime", {\em Performance Evaluations}, Vol. 65, n. 1-2, pp. 463--483, 2008.

\bibitem{AlonSpencer}N.~Alon, J.~Spencer, ``The Probabilistic Method", John Wiley and Sons, New York et al., 2000.%\vspace{-2mm}

\bibitem{Boldrini}C.~Boldrini, M.~Conti, A.~Passarella, ``ContentPlace: Social-Aware Data Dissemination in Opportunistic Networks", {\em Proc. ACM MSWiM}, pp. 203--210, 2008.%\vspace{-2mm}

\bibitem{Chain}A.~Chaintreau, P.~Hui, J.~Crowcroft, C.~Diot, R.~Gass, J.~Scott, ``Impact of Human Mobility on Opportunistic Forwarding Algorithms", {\em IEEE Transactions on Mobile Computing}, Vol. 6, n. 6, pp. 606--620, 2007.
%
\bibitem{Costa}P.~Costa, C.~Mascolo, M.~Musolesi, G.P.~Picco, ``Socially-Aware Routing for Publish-Subscribe in Delay-Tolerant Mobile Ad Hoc Networks", {\em IEEE Journal on Selected Areas in Communications}, Vol. 26, n. 5, pp. 748--760, May 2008.
%
\bibitem{Daly}E.~Daly, M.~Haahr, ``Social Network Analysis for Routing in Disconnected Delay-Tolerant MANETs", {\em Proc. ACM MobiHoc}, pp. 32--40, 2007.
%
\bibitem{Deza}M.M.~Deza, E.~Deza, {\em Encyclopedia of Distances}, Springer, Berlin, 2009.

%\bibitem{Diaz}J.~Diaz, A.~Marchetti-Spaccamela, D.~Mitsche, P.~Santi, J.~Stefa, ``Social-Aware Forwarding Improves Routing Performance in Pocket Switched Networks", {\em internet draft}, available at {\tt http://arxiv.org/abs????}
%
\bibitem{Erram}V.~Erramilli, M.~Crovella, A.~Chaintreau, C.~Diot, ``Delegation Forwarding", {\em Proc. ACM MobiHoc}, pp. 251--259, 2008.
%\bibitem{Gao}W.~Gao, Q.~Li, B.~Zhao, G.~Cao, ``Multicasting in Delay Tolerant Networks: A Social Network Perspective", {\em Proc. ACM MobiHoc}, 2009.
%\bibitem{David}H.A.~David, H.N.~Nagaraja, {\em Order Statistics}, John Wiley and Sons, 2003.

%\bibitem{Diaz}J.~Diaz, A.~Marchetti-Spaccamela, D.~Mitsche, P.~Santi, J.~Stefa, ``Social-Aware Forwarding Improves Routing Performance in Pocket Switched NetworksÓ, {\em Tech. Rep. IIT-11-2010}, Istituto di Informatica e Telematica del CNR, Pisa, 2010.
%
\bibitem{Fall}K.~Fall, ``A Delay-Tolerant Architecture for Challenged Internets", {\em Proc. ACM Sigcomm}, pp. 27--34, 2003.
%
\bibitem{Gao}W.~Gao, Q.~Li, B.~Zhao, G.~Cao, ``Multicasting in Delay Tolerant Networks: A Social Network Perspective", {\em Proc. ACM MobiHoc}, 2009.

\bibitem{Gross}M.~Grossglauser, D.N.C.~Tse, ``Mobility Increases the Capacity of Ad-Hoc Wireless Networks'', {\em Proc. IEEE Infocom}, pp. 1360--1369, 2001.
\bibitem{Groene}R.~Groenevelt, P.~Nain, G.~Koole, ``The Message Delay in Mobile Ad Hoc Networks'', {\em Performance Evaluation}, vol. 62, n. 1--4, pp. 210--228, 2005.
%
\bibitem{Hui}P.~Hui, J.~Crowcroft, E.~Yoneki, ``BUBBLE Rap: Social-Based Forwarding in Delay Tolerant Networks'', {\em Proc. ACM MobiHoc}, pp. 241--250, 2008.
%
\bibitem{Hui3}P.~Hui, A.~Chaintreau, J.~Scott, R.~Gass, J.~Crowcroft, C.~Diot, ``Pocket-Switched Networks and Human Mobility in Conference Environments", {\em Proc. ACM Workshop on Delay-Tolerant Networks (WDTN)}, pp. 244-251, 2005.
%
\bibitem{Ioannidis}S.~Ioannidis, A.~Chaintreau, L.~Massoulie, ``Optimal and Scalable Distribution of Content Updates over a Mobile Social Networks", {\em Proc. IEEE Infocom}, pp. 1422--1430, 2009.
\bibitem{Kara}T.~Karagiannis, J.-Y.~Le Boudec, M.~Vojnovic, ``Power Law and Exponential Decay of Inter Contact Times Between Mobile Devices", {\em Proc. ACM Mobicom}, pp. 183--194, 2007.
\bibitem{Li}F.~Li, J.~Wu, ``LocalCom: A Community-Based Epidemic Forwarding Scheme in Disruption-tolerant Networks", {\em Proc. IEEE Secon}, 2009.
%
\bibitem{McP}M.~McPherson, ``Birds of a feather: Homophily in Social Networks", {\em Annual Review of Sociology}, vol. 27, n. 1, pp. 415--444, 2001.

\bibitem{Mei}A.~Mei, J.~Stefa, ``SWIM: A Simple Model to Generate Small Mobile Worlds", {\em Proc. IEEE Infocom}, 2009.
%
\bibitem{Mei2}A.~Mei, G.~Morabito, P.~Santi, J.~Stefa, ``Social-Aware Stateless Forwarding in Pocket Switched Networks", {\em Proc. IEEE Infocom (miniconference)}, 2011.

%\bibitem{MitzenmacherU} M.~Mitzenmacher,  E.~Upfal, "Probability and Computing", Cambridge U.P., 2005.

\bibitem{Noulas}A.~Noulas, M.~Musolesi, M.~Pontil, C.~Mascolo, ``Inferring Interests from Mobility and Social Interactions, {\em Proc. ANLG Workshop}, 2009.
%
\bibitem{Resta}G.~Resta, P.~Santi, ``A Framework for Routing Performance Analysis in Delay Tolerant Networks with Application to Non-Cooperative Networks", {\em IEEE Trans. on Parallel and Distributed Systems}, to appear.
%
\bibitem{Spyro}T.~Spyropoulos, K.~Psounis, C.S.~Raghavendra, ``Efficient Routing in Intermittently Connected Mobile Networks: The Multi-copy Case'', {\em IEEE Trans. on Networking}, Vol. 16, n. 1, pp. 77--90, 2008.

\bibitem{Spyro2}T.~Spyropoulos, K.~Psounis, C.S.~Raghavendra, ``Efficient Routing in Intermittently Connected Mobile Networks: The Single-copy Case'', {\em IEEE Trans. on Networking}, Vol. 16, n. 1, pp. 63--76, 2008.
%\bibitem{Spyro3}T.~Spyropoulos, K.~Psounis, C.S.~Raghavendra, ``Performance Analysis of Mobility-assisted Routing'', {\em Proc. ACM MobiHoc}, pp. 49--60, 2006.
\bibitem{Vahdat}A.~Vahdat, D.~Becker, ``Epidemic Routing for Partially Connected Ad Hoc Networks'', {\em Tech. Rep. CS-200006}, Duke University, April 2000.%

\bibitem{Zhang}X.~Zhang, G.~Neglia, J.~Kurose, D.~Towsley, ``Performance Modeling of Epidemic Routing'', {\em Computer Networks}, Vol. 51, pp. 2867--2891, 2007.
\end{thebibliography}
\end{document}